\PassOptionsToPackage{dvipsnames}{xcolor}
\documentclass[sigplan,screen]{acmart}

\usepackage{cleveref}
\usepackage{thmtools}
\usepackage{thm-restate}

\usepackage{lipsum}

\usepackage{amscd,amsmath,multirow,amsthm}
\usepackage[justification=justified,skip=0pt]{caption}
\usepackage{enumitem}
\usepackage[makeroom]{cancel}

\newenvironment{stretchpars}
 {\par\setlength{\parfillskip}{0pt}}
 {\par}

\renewcommand{\paragraph}[1]{\vspace{0.35em}\noindent\textbf{#1}}

\usepackage{wrapfig}

\usepackage{mdframed}

\usepackage{comment}

\newcommand\remove[1]{}
\makeatother

\usepackage{tikz}

\newcommand{\hlcolor}{Yellow!35}
\newcommand{\hlcolorTwo}{LimeGreen!35}
\makeatletter
\newenvironment{btHighlight}[1][]
{\begingroup\tikzset{bt@Highlight@par/.style={#1}}\begin{lrbox}{\@tempboxa}}
{\end{lrbox}\bt@HL@box[bt@Highlight@par]{\@tempboxa}\endgroup}

\newcommand\btHL[1][]{%
  \begin{btHighlight}[#1]\bgroup\aftergroup\bt@HL@endenv%
}
\def\bt@HL@endenv{%
  \end{btHighlight}%   
  \egroup
}
\newcommand{\bt@HL@box}[2][]{%
  \tikz[#1]{%
    \pgfpathrectangle{\pgfpoint{1pt}{0pt}}{\pgfpoint{\wd #2}{\ht #2}}%
    \pgfusepath{use as bounding box}%
    \node[anchor=base west, fill=\hlcolor,outer sep=0pt,inner xsep=1pt, inner ysep=0pt, rounded corners=2pt, minimum height=\ht\strutbox+2pt,#1]{\raisebox{1pt}{\strut}\strut\usebox{#2}};
  }%
}

\newenvironment{btHighlightTwo}[1][]
{\begingroup\tikzset{bt@HighlightTwo@par/.style={#1}}\begin{lrbox}{\@tempboxa}}
{\end{lrbox}\bt@HLTwo@box[bt@HighlightTwo@par]{\@tempboxa}\endgroup}

\newcommand\btHLTwo[1][]{%
  \begin{btHighlightTwo}[#1]\bgroup\aftergroup\bt@HLTwo@endenv%
}
\def\bt@HLTwo@endenv{%
  \end{btHighlightTwo}%   
  \egroup
}
\newcommand{\bt@HLTwo@box}[2][]{%
  \tikz[#1]{%
    \pgfpathrectangle{\pgfpoint{1pt}{0pt}}{\pgfpoint{\wd #2}{\ht #2}}%
    \pgfusepath{use as bounding box}%
    \node[anchor=base west, fill=\hlcolorTwo,outer sep=0pt,inner xsep=1pt, inner ysep=0pt, rounded corners=2pt, minimum height=\ht\strutbox+2pt,#1]{\raisebox{1pt}{\strut}\strut\usebox{#2}};
  }%
}

\usepackage{listings}
\usepackage{etoolbox}
\newtoggle{InString}{}% Keep track of if we are within a string
\togglefalse{InString}% Assume not initally in string
\newcommand*{\ColorIfNotInString}[1]{\iftoggle{InString}{#1}{\color{blue}#1}}%
\newcommand*{\ProcessQuote}[1]{#1\iftoggle{InString}{\global\togglefalse{InString}}{\global\toggletrue{InString}}}%
\newcommand{\indentrule}{\color{code_indent}\vrule\hspace{2pt}}
\definecolor{code_indent}{HTML}{CCCCCC}

\newenvironment{figureAsListing}
    {
   \addtocounter{figure}{-1}
   \refstepcounter{lstlisting}

    \begin{figure}[!htbp]
        \centering
    }
    { 
        \end{figure} 
    }
    
\newenvironment{figureAsListingWide}
    {
    \addtocounter{figure}{-1}
   \refstepcounter{lstlisting}

    \begin{figure*}[!htbp]
        \centering
    }
    { 
        \end{figure*} 
    }

\usepackage{multicol}

\newcommand{\mynote}[3]{
    \fbox{\bfseries\sffamily\scriptsize#1}
{\small$\blacktriangleright$\textsf{\emph{\color{#3}{#2}}}$\blacktriangleleft$}}

\def\REVIEW{0}  %% Enables review commands 
\if \REVIEW 1
    \newcommand{\va}[1]{\mynote{Vitaly}{#1}{magenta}}
    \newcommand{\pk}[1]{\mynote{Petr}{#1}{blue}}
    \newcommand{\nk}[1]{\mynote{Nikita}{#1}{red}}
    \newcommand{\ap}[1]{\mynote{Anton}{#1}{cyan}}
\else
    \newcommand{\va}[1]{}
\newcommand{\pk}[1]{}
\newcommand{\nk}[1]{}
\newcommand{\ap}[1]{}
\fi

\author{Vitaly Aksenov}
\affiliation{%
\institution{City, University of London}
\city{London}
\country{UK}}
  
\author{Nikita Koval}
\affiliation{%
\institution{JetBrains}
\city{Amsterdam}
\country{The Netherlands}}

\author{Petr Kuznetsov}
\affiliation{%
\institution{T\'el\'ecom Paris, Institut Polytechnique de Paris}
\city{Paris}
\country{France}}

\author{Anton Paramonov}
\affiliation{%
\institution{EPFL}
\city{Lausanne}
\country{Switzerland}}


\begin{document}

\title{Memory Bounds for Concurrent Bounded Queues}

\setcopyright{none}
\renewcommand\footnotetextcopyrightpermission[1]{} % removes footnote with conference information in first column
\settopmatter{printacmref=false, printfolios=false}

\begin{abstract}
%A \emph{bounded FIFO queue} is a data structure that stores no more than the predefined capacity of elements and allows two operations: 1) insert an element, and 2) withdraw the oldest element. 
%
Concurrent data structures often require additional memory for handling synchronization issues in addition to memory for storing elements.   
Depending on the amount of this additional memory, implementations can be more or less \emph{memory-friendly}.
A \emph{memory-optimal} implementation enjoys the minimal possible memory overhead, which, in practice, reduces cache misses and unnecessary memory reclamation.

In this paper, we discuss the memory-optimality of non-blocking \emph{bounded queues}. Essentially, we investigate the possibility of constructing an implementation that utilizes a pre-allocated array to store elements and constant memory overhead, e.g., two positioning counters for \texttt{enqueue(..)} and \texttt{dequeue()} operations. Such an implementation can be readily constructed when the ABA problem is precluded, e.g., assuming that the hardware supports LL/SC instructions or all inserted elements are distinct. However, in the general case, we show that a memory-optimal non-blocking bounded queue incurs linear overhead in the number of concurrent processes. 
These results not only provide helpful intuition for concurrent algorithm developers but also open a new research avenue on the memory-optimality phenomenon in concurrent data structures.

\end{abstract}

\maketitle

%\vspace{-0.3cm}

\section{Introduction}
Developing concurrent algorithms is complex, and limiting synchronization is crucial in achieving high performance. Mutual exclusions allow building concurrent implementations on top of standard sequential algorithms at the expense of scalability. \textit{Non-blocking} synchronization has the potential to improve scalability, but it often requires allocating extra memory. The classic Michael-Scott queue~\cite{MS96} serves as an illustrative example, allocating a new node for each enqueue and, therefore, using a significant amount of memory for node objects; this design results in increased memory usage and cache misses. Modern queues~\cite{gidenstam2010cache,morrison2013fast,yang2016wait} incorporate several elements into a single node, significantly reducing memory allocation and creating a more compact memory structure that mitigates cache misses. We consider these solutions more \emph{memory-friendly}, as they allocate less additional memory. The intuition is that the more memory-friendly the implementation is, the better and more predictable performance it usually provides under both high and low contention. This is especially important for general-purpose solutions in standard libraries of programming languages and popular frameworks.

In this paper, we raise the question of \emph{memory-optimality}: \textbf{what is the absolute minimum of memory overhead a concurrent data structure must incur?} 
We narrow our attention to \emph{bounded queues}, ubiquitous in resource management systems and task schedulers, e.g., \texttt{io\_uring} in Linux kernel~\cite{iouring}, or high-speed networking and storage libraries such as DPDK~\cite{DPDK} and SPDK~\cite{SPDK}.
%bounded queues present a valuable subject of study. 
%
Bounded queues allow for a straightforward definition of memory overhead: the amount of memory that must be allocated on top of the fixed memory required for storing the queue elements.

\setlength{\columnsep}{1em}%
\begin{wrapfigure}[9]{r}{0.38\linewidth}
\begin{center}
    \vspace{-1.3em}
    \includegraphics[width=1\linewidth]{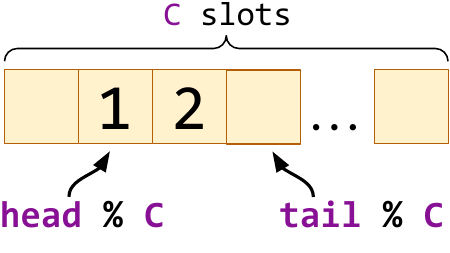}
    % \vspace{0.01em}
    \caption{Sequential bounded queue on top of an array of $C$ slots and two counters.}
    \label{fig:sequential}
    \vspace{0.2em}
\end{center}
\end{wrapfigure}

% \paragraph{Sequential bounded queue.}
A straightforward sequential bounded queue implementation maintains an array of size $C$ (the queue capacity) equipped with two counters that track the total numbers of \texttt{enqueue(..)} and \texttt{dequeue()} invocations; their values modulo $C$ point to the next working slots of the operations. See Figure~\ref{fig:sequential} for an illustration. This implementation utilizes exactly $C$ array slots to store elements and two memory locations for the counters, producing the $\Theta(1)$ memory overhead.

Can we build a \emph{concurrent} bounded queue with the same memory footprint?
A trivial solution that utilizes coarse-grained locking incurs constant overhead required for the lock but is inefficient under high load. 
To the best of our knowledge, most fast and scalable concurrent queues are \emph{non-blocking}, e.g., \cite{morrison2013fast,yang2016wait,LPRQ,nikolaev2019scalable}. 
While the classic Michel-Scott queue~\cite{MS96} algorithm allocates a new node for each element, these modern solutions pack multiple elements into a single node, making the implementation more memory-friendly. 
Nevertheless, the memory overhead remains linear in the number of elements. 

For \emph{bounded} queues, it is natural to utilize a pre-allocated array of $C$ slots to store elements. 
This raises an interesting question: is it possible to design an algorithm where the memory overhead depends on neither the queue capacity nor the number of concurrent processes? 
Can we architect a non-blocking bounded queue using just an array for elements and two positioning counters for \texttt{enqueue(..)} and \texttt{dequeue()}, similarly to the sequential implementation?

\vspace{-0.1em}
\paragraph{Our contribution.} 
We first observe that the primary challenge in building a concurrent bounded queue is the ABA problem~\cite{herlihy2020art}.
%
%Intuitively, it may lead to \texttt{dequeue()} removing an element from the middle of the queue, assuming it is the first one. 
%
When the ABA problem is precluded, e.g., when all elements are distinct or using the \texttt{LL/SC} synchronization primitives instead of \texttt{CAS}, we show the possibility of designing an algorithm with \emph{constant} memory overhead (Section~\ref{sec:special-cases}).  

% \noindent
In the general case, we prove that any \emph{obstruction-free}~\cite{OF}\footnote{Obstruction-freedom is the weakest non-blocking progress condition, which guarantees progress in any chosen thread when all the others are paused~\cite{OF}.} bounded queue implementation must use additional $\Omega(T)$ memory locations for synchronization, where $T$ is the number of processes (Section~\ref{sec:lb}). 
To show that, we construct a non-linearizable execution for any algorithm that utilizes fewer memory locations. 
The lower bound appears to be tight: we present an algorithm with $O(T)$ memory overhead.

\vspace{-0.1em}
\paragraph{Practical impact.}
In our industrial experience, we witnessed numerous attempts to design a concurrent bounded queue with constant overhead (e.g., on a pre-allocated array with positioning counters). All of these attempts resorted to practical trade-offs like periodic memory allocation, blocking behavior, or relaxed semantics. Our results inform the practitioners that these trade-offs are unavoidable, potentially saving a tremendous amount of time and mental energy for those thinking otherwise.  

Notably, most modern \emph{unbounded} concurrent queues pack multiple elements into fixed-capacity segments and try to reuse them~\cite{LPRQ,morrison2013fast,nikolaev2019scalable}. 
These segments, also known as \emph{ring buffers}, essentially are bounded queues with slightly relaxed semantics, allowing \texttt{enqueue(..)} to fail spuriously and close the segment for further additions. 
To reuse the segments, they equip each slot with a 64-bit epoch value, leading to $\Theta(C)$ memory overhead, where $C$ is the segment capacity. 
%
%Our results do not prove but provide an intuition that it is impossible to implement such a reusable ring buffer structure with constant memory overhead.

Notice that our bounds do not necessarily imply that ring buffers, widely used in real systems, are impractical. However, these ring buffers should relax the semantics, relax the progress guarantee, or employ non-constant memory overhead. While one may \emph{share} this overhead between \emph{multiple} bounded queues, our primary result is the impossibility of achieving constant overhead, which significantly affects the algorithm design and its trade-offs.
%
%They are usually faster than linked-list based queues due to the reuse of memory.
%
% should such a buffer be implemented using standard synchronization primitives. 
%

%in the case when the application uses multiple bounded queues, we can achieve $O(n)$ memory overhead \emph{in total} by sharing the synchronization mechanism be shared between them.

\vspace{-0.1em}
\paragraph{Theoretical impact.}
This work opens a wide research avenue on memory-friendly and memory-optimal concurrent computing. 
It sets strict bounds on memory overhead in dynamic data structures and provides insights on the optimal memory overhead their implementations can achieve. 
We believe that the proposed theoretical framework can be extended to other bounded and unbounded data structures, leading to more practical implications.

\vspace{-0.2em}
\section{Memory Overhead: The Intuition}\label{sec:special-cases}
\vspace{-0.1em}

%[[PK sounds a bit repetitive
%Intuitively, the memory overhead incurred by a concurrent implementation of a bounded queue is the amount of memory required in addition to the memory used to store the values.
%]
%
We start by developing intuition on memory-friendliness and the challenges in attaining constant memory overhead. To illustrate the latter, we present solutions that work under specific assumptions, thereby shedding light on the fundamental obstacles. 
% on memory-friendliness in general and why constant memory overhead is hard to achieve by presenting solutions that work under specific assumptions.

% In this section, we consider the memory-optimality question under certain popular restrictions either on the application or on the model. %exhibiting constant or $\Theta(n)$ memory overhead. 
%
%We believe that these algorithms are useful, as the restrictions are quite common in practice.
%
%The restricted constant-overhead algorithms gradually bring us closer to our memory-optimal algorithm that achieves $\Theta(n)$ memory overhead in the model without restrictions.
%
%This algorithm is presented in the next section.%~\ref{sec:upper-bound}.

\vspace{-0.3em}
\subsection{Simplest Memory-Friendly Queue}
\vspace{-0.1em}

We start with a \emph{memory-friendly} bounded queue algorithm based on the standard linked-list queue. 
While it is not memory-optimal, it can be tuned to be more or less memory-friendly, thus developing an intuition of memory-friendliness. 
Note that further we consider implementations that use a more suitable container~--- the circular buffer.

The idea is to build a bounded queue on a conceptually infinite array with \texttt{head} and \texttt{tail} counters directly pointing to the array slots. Similar to the Michael-Scott queue design~\cite{MS96}, \texttt{enqueue(..)} first installs the element in the next empty slot, followed by the \texttt{tail} counter increment. Similarly, \texttt{dequeue()} begins by extracting the first element, incrementing its \texttt{tail} counter after that. Listing~\ref{lst:arr} presents the algorithm. As each \texttt{enqueue-dequeue} pair manipulates a unique array slot, the same element cannot be installed into the same slot, so the ABA problem is naturally eliminated.

\paragraph{Infinite array.} 
We suggest implementing the infinite array as a concurrent linked list of fixed-size segments, following the approach proposed for synchronous queues and channels in Kotlin Coroutines~\cite{koval2019scalable,CHANNELS_PPOPP23}. All segments are marked with a unique \texttt{id} and follow each other; Figure~\ref{fig:inf_arr} illustrates the high-level structure. To access the \texttt{i}-th cell in the infinite array, we find
the segment with \texttt{id\:=\:i\:/\:K} (adding it to the linked list if needed) and go to the cell \texttt{i\:\%\:K} in it. 

\begin{figure}[H]
\centering
% \vspace{-1em}
\includegraphics[width=0.9\linewidth]{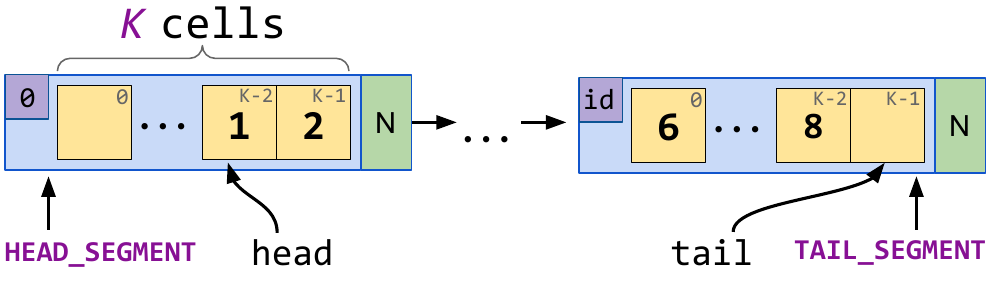}
% \vspace{0.1em}
\caption{High-level structure of the memory-friendly bounded queue based on a conceptually infinite array implemented via a concurrent linked list of fixed-size segments.}
\vspace{-0.4em}
\label{fig:inf_arr}
\end{figure}

\begin{figureAsListing}
% \vspace{-1.5em}
\begin{lstlisting}
fun enqueue(e:#\:#E): Bool = while(true)#\:#{
 // Read the counters snapshot
 t := tail; h := head #\label{line:arr_enq_counters_read_start}#
 if t != tail: continue #\label{line:arr_enq_counters_read_end}#
 // Is the queue full?
 if t == h + C: return false #\label{line:arr_enq_check_fullness}#
 // Try to insert the element
 done := CAS(&a[t], #$\perp$#, e) #\label{line:arr_enq_cas}#
 // Increment the counter
 CAS(&tail, t, t + 1) #\label{line:arr_enq_inc_counter}#
 // Finish on success
 if done: return true 
}
fun dequeue(): E? = while(true) {

 // Read the counters snapshot
 t := tail; h := head #\label{line:arr_deq_counters_read_start}#
 if t != tail: continue #\label{line:arr_deq_counters_read_end}#
 // Is the queue empty?
 if h == t: return #$\perp$# #\label{line:arr_deq_empty}#
 // Try to extract an element 
 e := a[h] #\label{line:arr_deq_mod1}# 
 done := e != #$\perp$# && CAS(&a[h], e, #$\perp$#)#\label{line:arr_deq_mod2}#
 // Increment the counter
 CAS(&head, h, h + 1) #\label{line:arr_deq_inc}#
 // Finish on success
 if done: return e
}
\end{lstlisting}
% \vspace{-1em}
% \vspace{-0.3em}
\caption{Memory-friendly bounded queue of capacity \texttt{C} implemented on an infinite array.} %The difference from the baseline algorithm in Figure~\ref{lst:ppopp} is highlighted with yellow.\pk{PK: should it be Fig. 1 here?}}
% \vspace{-1.5em}
\label{lst:arr}
\end{figureAsListing}

% \pagebreak
\paragraph{Memory overhead.}
Each segment has a memory overhead of $\Theta(1)$ to store its \texttt{id} and a pointer to the next segment. However, we still have to address the memory reclamation problem: once a segment no longer contains queue elements, it should be recycled, and the memory it occupies should be released. To make the memory overhead bounded, we suggest reusing segments by applying the technique to reclaim descriptors~\cite{tbrown_recyclable_desc}. This approach requires $\Theta(T)$ additional segments on top of the ones that store values, where $T$ is the number of concurrent processes. With storing $K$ elements in each segment, the overall solution exhibits $\Theta(C/K\:+\:T\cdot K)$ memory overhead.
Though not memory-optimal, this simple algorithm can be tuned to be more or less memory-friendly by adjusting the segment size $K$, achieving the minimum of $\Theta(T \cdot \sqrt{C})$ when choosing $K = \sqrt{C}$.

%
% Here we maintain a Michael-Scott linked list of indexed segments of a fixed size $K$. 
%

%
% By that, the position $i$ in the queue corresponds to a cell ($i$ modulo $K$) in the segment with the identifier $\lfloor i/K \rfloor$.
%
%
%We illustrate this structure in Figure~\ref{fig:inf_arr}. 

%

%
% To achieve the finite memory overhead, .
%

%\pk{PK: Very sketchy: we never talked about descriptors before. How is "recycling" done?  why at most $n$? Recycling allows for simulating infinite arrays? Not very clear...}

% With storing $K$ elements in each segment, the algorithm exhibits $\Theta(C/K\:+\:K\:\cdot\:T)$ memory overhead: $\Theta(C/K)$ comes from the pointers of the linked list and $\Theta(K \cdot T)$ from the memory reclamation algorithm for the implementation of the infinite array. 
% %
%Such an algorithm suits well small-contention scenarios. 
%

% Moreover, it respects our separation of memory into value- and metadata-locations.
%
%However, in case of high contention and when the capacity $C$ is relatively small, it is better, in our opinion, to use the bounded queue algorithm presented by Nikolaev~\cite{nikolaev2019scalable}.
%
% However, the queue capacity $C$ (even divided by the segment size $K$) may be significantly larger than the number of processes $T$. 
% Therefore, ideally, we would like to ensure that the overhead does not depend on $C$.
%this algorithm does not satisfy us since it uses more than $\Theta(C + n)$ amount of additional memory, while our main goal is to design an algorithm with the memory overhead that depends only on $n$: $C$ is usually larger than $n$.

\subsection{Constant Overhead with Distinct Elements} \label{subsec:ppopp}
\begin{stretchpars}
We now describe an algorithm that uses only $O(1)$ additional memory under two reasonable assumptions. 
First, we require all inserting elements to be distinct, which is common in practice.
%
% It is common to use queues for storing uniquely identified tasks or identifiers themselves.
%
Second, the system should be provided with an unlimited supply of versioned $\perp$ (\texttt{null}) values, so we can replace unique elements with unique \texttt{$\perp$}-s on extractions. 
This can be achieved by stealing one bit from addresses (values) to mark them as \texttt{$\perp$}-s and using the rest to store the $\perp$ version. 
These assumptions help eliminate the \emph{ABA problem}. 

Essentially, both \texttt{enqueue(..)} and \texttt{dequeue()} (1)~read the counters in a snapshot manner, (2)~try to perform a ``round-valid'' update (\texttt{enqueue(..)} replaces $\bot_\mathtt{round}$ with the element while \texttt{dequeue()} does the opposite, where \texttt{round = counter\:/\:C}), and (3)~increase the counter via \texttt{CAS}.
The pseudocode is presented in Listing~\ref{lst:ppopp}.

\end{stretchpars}

\begin{lstlisting}[
caption={
Bounded queue algorithm with $O(1)$ additional memory that requires elements to be distinct and an unlimited supply of versioned $\perp$ (null) values.
},
label={lst:ppopp}
]
fun enqueue(e:#\:#E): Bool = while(true)#\:#{
 // Read the counters snapshot
 t := tail; h := head #\label{line:ppopp_enq_counters_read_start}\label{line:ppopp_enqueue_lin}#
 if t != tail: continue #\label{line:ppopp_enq_counters_read_end}#
 // Is the queue full?
 if t == h + C: return false  #\label{line:ppopp_enq_check_fullness}#
 // Try to insert the element
 round := t / C; i := t#\:#%#\:#C
 done := CAS(&a[i], #$\perp_{\mathtt{round}}$#, e)#\label{line:ppopp_enq_cas}#
 // Increment the counter
 CAS(&tail, t, t + 1)  #\label{line:ppopp_enq_inc_counter}#
 // Finish on success
 if done: return true 
}

fun dequeue(): E? = while(true) {
 // Read the counters+element snapshot
 t := tail; h := head; e := a[h#\:#%#\:#C]#\label{line:ppopp_deq_counters_read_start}\label{line:ppopp_deq_lin}#
 if t != tail: continue #\label{line:ppopp_deq_counters_read_end}#
 // Is the queue empty?
 if t == h: return #$\perp$#  #\label{line:ppopp_deq_empty}#
 // Try to extract the element
 round := h#\:#/#\:#C + 1; i := h#\:#%#\:#C
 done := e#\:#!=#\:##$\perp_{\mathtt{round}}$# && CAS(&a[i],#\:#e,#\:##$\perp_{\mathtt{round}}$#)#\label{line:ppopp_deq_cas}\label{line:ppopp_deq_check_null}#
 // Increment the counter
 CAS(&head, h, h + 1)  #\label{line:ppopp_deq_inc}#
 // Finish on success
 if done: return e
}
\end{lstlisting}

\subsection{Constant Overhead with LL/SC}\label{subsec:llsc}
% In the previous subsection we assumed that we have distinct elements and an infinite supply of nulls in order to avoid the ABA problem.
%
Another way to avoid the ABA problem is to use the \texttt{LL/SC} (\emph{load-link/store-conditional}) ABA-immune synchronization primitives instead of \texttt{CAS}.
In this algorithm (Listing~\ref{lst:llsc}), both \texttt{enqueue(..)} and \texttt{dequeue()} begin by reading a snapshot of the counters and the corresponding array slot {---} this part remains unchanged. However, to update the slot state, we now use the \texttt{SC} primitive instead of \texttt{CAS}. By inserting the \texttt{LL} instruction before reading the snapshot, we ensure that the update fails if the cell state has been changed.

The algorithm shows that the \texttt{LL/SC} primitives are conceptually more powerful than \texttt{CAS}. The \texttt{LL/SC} instructions, though not widely available in programming languages, can be used in ARM, PowerPC, RISC-V, and some other architectures (with a certain risk of spurious failures).

\subsection{$\Theta(T)$ Overhead with DCSS}\label{subsec:dcss}
Another way to eliminate the ABA problem is to use the \texttt{Double-Compare-Single-Set} (\texttt{DCSS}) synchronization primitive. \texttt{DCSS(\&A, expectedA, updateA, \&B, expectedB)}

% \vspace{1em}
\begin{lstlisting}[
caption={
Bounded queue algorithm with $O(1)$ additional memory via \texttt{LL/SC}. This is a modification of the algorithm in Listing~\ref{lst:ppopp}; the changes are highlighted.
},
label={lst:llsc}
]
fun enqueue(e:#\:#E): Bool = while(true)#\:#{
 // Read the counters snapshot
 t := tail; h := head  #\label{line:llsc_enq_counters_read_start}#
 ##@state := LL(&a[t % C])@
 if t != tail: continue #\label{line:llsc_enq_counters_read_end}#
 // Is the queue full?
 if t == h + C: return false #\label{line:llsc_enq_check_fullness}#
 // Try to insert the element
 ##@done := state#\:#==#\:##\btHL$\perp$# && SC(&a[t#\:#%#\:#C],#\:#e)@#\label{line:llsc_enq_sc}#
 // Increment the counter
 if LL(&tail) == t: SC(&tail, t#\:#+#\:#1) #\label{line:llsc_enq_inc}#
 // Finish on success
 if done: return true
}

fun dequeue(): E = while(true) {
 // Read the counters+element snapshot
 h := tail; h := head  #\label{line:llsc_deq_counters_read_start}#
 ##@e := LL(&a[h % C])@
 if t != tail: continue #\label{line:llsc_deq_counters_read_end}#
 // Is the queue empty?
 if t == h: return #$\perp$# #\label{line:llsc_deq_empty}#
 // Try to extract the element
 ##@done := e != #\btHL$\perp$# && SC(&a[h#\:#%#\:#C], #\btHL$\perp$#)@ #\label{line:llsc_deq_sc}#
 // Increment the counter
 if LL(&head) == d: SC(&head, h#\:#+#\:#1) #\label{line:llsc_deq_inc}#
 // Finish on success
 if done: return e 
}
\end{lstlisting}
% \vspace{0.5em}

\noindent
checks that the values located by addresses \texttt{A} and \texttt{B} are equal to \texttt{expectedA} and \texttt{expectedB}, respectively, updating \texttt{A} to \texttt{updateA} and returning \texttt{true} if the check succeeds, and returning \texttt{false} otherwise.

% In the full version of the paper (supplementary material, Listing~3) we present the pseudocode of the algorithm. We use \texttt{DCSS} to atomically update the cell and check that the corresponding counter has not been changed. If \texttt{DCSS} fails, the algorithm helps to increment the counter and restarts the operation. The rest is the same as in the previous algorithms: getting an atomic snapshot of the counters (and the cell state for \texttt{dequeue()}), checking whether the operation is legal (fullness check for \texttt{enqueue(..)} and emptiness one for \texttt{dequeue()}), trying to update the cell, and incrementing the corresponding counter at the end.

We use \texttt{DCSS} to atomically update the slot and check that the corresponding counter has not been changed. If \texttt{DCSS} fails, the algorithm helps to increment the counter and restarts the operation. The rest is the same: getting a counters snapshot (including the slot state for \texttt{dequeue()}), checking that the queue is not full/empty, and trying to update the slot state, incrementing the operation counter at the end. Listing~\ref{lst:dcss} presents the pseudocode.

For the implementation of \texttt{DCSS} we can use a \emph{descriptors} approach presented in~\cite{harris2002practical}. In short, each \texttt{DCSS} call creates a descriptor, which describes the operation, and installs it in the updating location, thus, preventing updates while reading the second location and allowing other threads to help complete the operation. While the naive implementation requires allocating a descriptor object on each \texttt{DCSS} call, it is possible to recycle these descriptors so that only $2 T$ of them are required~\cite{tbrown_recyclable_desc}, thus, incurring $\Theta(T)$ additional memory.

\begin{lstlisting}[
caption={
Bounded queue algorithm with $O(T)$ overhead via recyclable \texttt{DCSS} descriptors. This is a modification of the algorithm in Listing~\ref{lst:ppopp}, the changes are highlighted.
},
label={lst:dcss}
]
fun enqueue(e:#\:#E): Bool = while(true)#\:#{
 // Read the counters snapshot
 t := tail; h := head #\label{line:dcss_enq_counters_read_start}#
 if t != tail: continue #\label{line:dcss_enq_counters_read_end}#
 // Is the queue full?
 if t == h + C: return false #\label{line:dcss_enq_check_fullness}#
 // Try to insert the element
 ##@done := DCSS(&a[t#\:#%#\:#C],#\:##\btHL$\perp$#,#\:#e,#\:#&tail,#\:#t)@ #\label{line:dcss_enq_dcss}#
 // Increment the counter
 CAS(&tail, t, t + 1) #\label{line:dcss_enq_inc_counter}#
 // Finish on success
 if done: return true
}

fun dequeue(): E = while(true) {
 // Read the counters+element snapshot
 t := tail; h := head #\label{line:dcss_deq_counters_read_start}#
 e := @DCSS_Read(a[h#\:#%#\:#C])@
 if t != tail: continue #\label{line:dcss_deq_counters_read_end}#
 // Is the queue empty?
 if t == h: return #$\perp$# #\label{line:dcss_deq_empty}#
 // Try to extract the element
 ##@done#\:#:=#\:#DCSS(&a[h#\:#%#\:#C],#\:#e,#\:##\btHL$\perp$#,#\:#&head,#\:#h)@ #\label{line:dcss_deq_dcss}#
 // Increment the counter
 CAS(&head, h, h#\:#+#\:#1)  #\label{line:dcss_deq_inc}#
 // Finish on success
 if done: return e
}
\end{lstlisting}
\vspace{1em}

% \noindent

% In short, a descriptor is a special object that ``announces'' an operation, in this case, DCSS, which allows it to be \emph{helped} by other processes.
%[[PK we cannot say "matches", as the model is different
% that matches the lower bound. 
%]]

% \noindent
% Assuming the queue stores references, the distinction between \emph{values} (enqueued elements) and descriptors can be implemented by either stealing a marker bit from addresses or using a language construction like \texttt{Variant} in Rust or \texttt{instanceof} in Java to check whether the object is a descriptor or a value (for example, \texttt{Integer}).

% Note that the descriptor-based implementation of DCSS requires the ability to store values and descriptors in the same array cells. %, leading to the smaller universe of values.
% In practice, this solution can be used for a bounded queue of references in the most popular languages.
% For more information we point to Appendix~\ref{subsec:special-cases-full:dcss}.

\vspace{-0.3em}
\subsection{Values and Metadata}
% \vspace{-0.2em}

Notably, the algorithms with the ``distinct elements'' assumption (Subsection~\ref{subsec:ppopp}) and \texttt{DCSS} (Subsection~\ref{subsec:dcss}) require an ability to store elements and metadata (distinct $\bot$-s and \texttt{DCSS} descriptors) in the same array slots.
In practice, to distinguish metadata from data, we either have to steal a bit from values, which results in a smaller universe of possible values, or use language constructs that allow distinguishing types but might lead to a linear memory overhead, such as storing \textsf{Object}s in Java (which incurs boxing when storing primitive values) or using \texttt{Variant} in Rust (which is essentially a wrapper). 
%Besides, many programming languages do not allow storing primitive and \texttt{Object} types in the same memory locations, let alone bit stealing.
% Our main target is to avoid these drawbacks.
%
%Thus, both solutions lead to a larger memory overhead.
%

% In the next section, we prove the lower bound on the additional memory and overview an algorithm that matches this bound asymptotically and separates the memory for values and metadata.

%%%%%%%%%%%%%%%%%%%%%%%%%%%%%%%%%%%%%%%%%%%%%%%%%%%

\vspace{0.4em}
% \pagebreak
\section{Memory Overhead: The Lower Bound}\label{sec:lb}
% \vspace{-0.2em}
Recall that the memory overhead of a bounded queue refers to the extra memory necessary beyond that used to store the elements. We now prove that the minimal memory overhead of a linearizable non-blocking bounded queue is linear in the number of concurrent threads and show that this lower bound is tight by presenting a matching algorithm.

% \vspace{-0.3em}
\subsection{The Lower Bound: Proof Overview}
\label{subsec:proof-overview}
% \vspace{-0.2em}

We proceed by contradiction.
Suppose that there exists an obstruction-free linearizable bounded queue of capacity $C$ that employs $X < \frac{T}{24} - 1$ value-locations beyond the $C$ memory locations to store the queue elements; $T$ is the number of concurrent threads.
(``$24$'' here is a constant that simplifies the proof without affecting the desired $\Theta(T)$ lower bound.)
We construct a non-linearizable execution with such an implementation, showing that any obstruction-free linearizable bounded queue must incur $\Theta(T)$ memory overhead. 

To construct a non-linearizable execution, we first catch $X + 3$ threads (recall that $X$ is the memory overhead) in a state when they are about to update different memory locations and pause these threads. For simplicity, we assume that these threads will perform writes to the memory locations and never update them to \texttt{null} ($\bot$); we explain how to support updates via \texttt{CAS} and \texttt{null} values later. 

Then, carefully ensuring that the queue is empty, we successfully add $C$ elements $x_1, x_2, ..., x_C$ into the queue {---} following \texttt{dequeue()}s would have extracted these elements in the same order.
However, we do not extract the elements now. Instead, we choose a thread caught right before updating a memory location, which stores an element $x_i$ where $1 < i < C$, and is about to put $y$ there. Intuitively, this write will remove $x_i$ from the middle of the queue and put $y$ instead. As we have caught $X + 3$ threads, and the queue uses $X$ additional value-locations, at least three threads are about to update the location where the queue elements are stored and one of them must point to a middle element. Figure~\ref{fig:proof1} illustrates the resulting state. 

\begin{figure}[H]
    \centering
    \vspace{-1em}
    \includegraphics[width=0.85\linewidth]{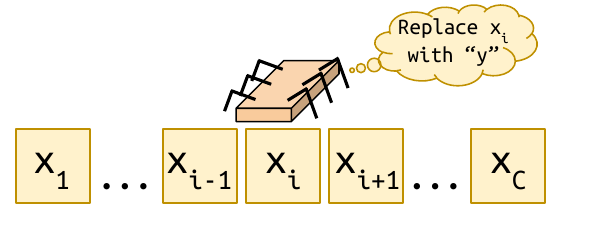}
    % \vspace{4pt}
    \caption{A thread is about to remove $x_i$ from the middle of the queue and put $y$ instead.}
    \vspace{-0.5em}
    \label{fig:proof1}
\end{figure}

We now replace the value $x_i$ with $y$ by invoking ``write'' in the caught thread. Assuming the bounded queue implementation does not differentiate elements, the state is now the same as when filling the queue with $x_1, x_2, ..., x_{i-1}, y, x_{i+1}, ..., x_C$ elements. Therefore, the following $C$ \texttt{dequeue()}-s successfully extract these elements. 

Intuitively, this execution is non-linearizable, as we have managed to put an element in the middle of the queue, observing ``old'' elements $x_1, x_2, ..., x_{i-1}$, then a ``new'' element $y$, which could not be installed concurrently as the queue is bounded, and followed by ``old'' elements $x_{i+1}, ..., x_C$. 

% Assume we have an execution that successfully adds $C$ elements $v_1, v_2, ..., v_C$ into the queue {---} following \texttt{dequeue()}s would have extracted these elements in the same order. We now extend this execution by carefully catching  right before they perform \texttt{CAS($x_i$, $y_i$)} or \texttt{WRITE($y_i$)} in different memory locations. As the queue incurs $X$ additional memory locations, $3$ of these threads manipulate the memory locations where $v_1, v_2, ..., v_C$ are stored, and one of these threads manipulates with a location related to a middle element of the queue. See Figure~\ref{} for an illustration. 

% In case the thread that points to a middle element of the queue performs \texttt{CAS(..)}, we change the execution in a way so this \texttt{CAS(..)} expects $v_i$ and is about to succeed. Then, we apply this CAS, thus, putting a new element $y$ into the middle of the queue.

% Finally, we extract all the elements and observe that. 

Note that this is a very simplified version of the proof dedicated to creating an intuition on constructing a  non-linearizable execution. To complete the proof, we need to be extremely careful with catching the threads, extending executions, and supporting \texttt{CAS} operations and \texttt{null} values. The following subsections describe the model we use, the reasonable algorithm assumptions that simplify the proof, and the full proof itself.

\subsection{Model}
\label{sec:model}
\vspace{-0.2em}

\paragraph{Threads and atomic operations.}
We consider a system of $T$ processes (threads) $p_1, \ldots, p_T$ communicating
by invoking \emph{primitives} on shared memory. Besides plain reads and writes, we assume \texttt{Compare-and-Set} primitive, denoted as $\texttt{CAS(\&a, old, new)}$, that atomically checks the value at address $a$ and, if it equals $\texttt{old}$, replaces it with \texttt{new} and returns \texttt{true}, returning \texttt{false} otherwise.

While we do not support \texttt{Fetch-and-Add} implicitly, it does not break the proof when not manipulating memory locations where elements are stored. In particular, using \texttt{Fetch-and-Add} for updating counters is allowed.

The \texttt{LL/SC} (load-link/store-conditional) synchronization primitives, however, are not allowed, as they are ABA-immune and enable bounded queue implementation with constant memory overhead, as discussed in Subsection~\ref{subsec:llsc}.

\paragraph{Bounded queue data structure.}
The \emph{Bounded Queue} (BQ) data structure parameterized with type \texttt{Type} and capacity $C$ stores elements of type \texttt{Type} and provides two operations:
\begin{itemize}[noitemsep,topsep=0pt]
    \item \texttt{enqueue($x$: Type)} (or \texttt{enq($x$)}): if the size of the queue is less than $C$, the operations adds value $x$ and returns \texttt{true}; otherwise, the operation returns \texttt{false}.
    \item \texttt{dequeue()} (or \texttt{deq()}) retrieves the oldest element from the queue, returning $\perp$ if the queue is empty.
\end{itemize}
%[[PK do not see why it is relevant
% \remove{
% Note that by the definition, we require the existence of a special value $\perp$ of type \texttt{Type}. Its existence is clear when \texttt{Type} is \texttt{Object} as we can set $\perp$ to \texttt{null}. However, when \texttt{Type} is some basic type, such as \texttt{Integer}, \texttt{Double} and etc., we require to choose the special value, that cannot be inserted, for example, $0$.
% }
% \noindent

Here, $\perp$ is a special \texttt{null} value, available in most programming languages for references. Also, we omit type \texttt{Type} when discussing BQ.

% A high-level concurrent object or a data type is a tuple $(\Phi, \Gamma, Q, q_0, \theta)$,
% where $\Phi$ is a set of operations, $\Gamma$ is a set of responses, $Q$ is a set of states,
% $q_0$ is an initial state and a transition function $\theta \subset Q \times \Phi \times Q \times \Gamma$,
% that determines, for each state and each operation, the set of possible resulting
% states and produced responses.

% In this paper we concentrate on the data structures that satisfy the following data type.

%\begin{definition}[Implementation]
% \paragraph{High-Level Object Implementation.}
\paragraph{Implementation.}
Informally, an \emph{implementation} of a bounded queue $BQ$ is a distributed algorithm $A$ consisting of automata $A_1, \ldots, A_T$. 
Given an operation invoked on $BQ$ by $p_i$, $A_i$ specifies the sequence of primitive operations on the shared memory that $p_i$ needs to execute to compute a response.
%\end{definition}
An \emph{execution} of an implementation of $BQ$ is a sequence of \emph{events}: invocations and responses of high-level operations on $BQ$, as well as primitives applied by the processes on the shared-memory locations and the responses they return so that the automata $A_i$ and the primitive specifications are respected.
An operation is \emph{complete} in a given execution if a matching response follows its invocation.  
An incomplete operation can be \emph{completed} by inserting a matching response after invocation.
We only consider \emph{well-formed} executions: a process never invokes a new high-level operation before its previous operation returns.
An operation $o_1$ \emph{precedes} operation $o_2$ in execution $\alpha$ (we write $o_1\preceq_{\alpha} o_2$)  if the response of $o_1$ precedes the invocation of $o_2$ in $\alpha$.           

For simplicity, we only consider \emph{deterministic} implementations: all shared-memory primitives and automata $A_1,\ldots,A_T$ exhibit deterministic behavior.  
Our lower bound can be, however, easily extended to randomized implementations: in the worst case, we can suggest that a random uses some deterministic algorithm leading to the deterministic execution.

\paragraph{Concurrent execution.}
A concurrent execution $\alpha$ of a bounded queue is \emph{linearizable} ~\cite{herlihy1990linearizability,AW04} if all complete operations in it and a subset of incomplete ones can be put in a total order $S$ such that (1) $S$ respects the sequential specification of the bounded queue, and (2)  $\preceq_{\alpha} \subseteq \preceq_S$, i.e., the total order respects the \emph{real-time} order on operation in $\alpha$.  
Informally, every operation $o$ of $\alpha$ can be associated with a \emph{linearization point} put within the operation's \emph{interval}, the fragment of $\alpha$  between the invocation of $o$ and the response of $o$.   
An implementation is \emph{linearizable} if any its execution is linearizable.
%
% \remove{
% linearization points can be selected for each completed operation and for a subset of the operations that started but did not complete, such that the linearization point for an operation occurs during the operation, and the result of each completed operation in $\alpha$ is the same as it would be if the operations were executed atomically at their linearization points.
% }

\paragraph{Progress conditions.}
In this paper, we focus on \emph{non-blocking} implementations that, intuitively, do not involve any form of locking: a failure of a process does not prevent other processes from making progress.
We can also talk about non-blocking implementations of specific operations.

Popular liveness criteria are lock-freedom and obstruction-freedom~\cite{herlihy2020art}. 
An implementation of an operation is \emph{lock-free} if, in every execution, at least one process is guaranteed to complete every   
such operation it invokes in a finite number of its own steps. 
An implementation of an operation is \emph{obstruction-free} if it guarantees that if a process invokes this operation and, from some point on, takes sufficiently many steps \emph{in isolation}, i.e., without contending with other processes, then the operation eventually competes.
We say that an implementation is obstruction-free (resp., lock-free) if all its operations are obstruction-free (resp., lock-free).
Lock-freedom is a strictly stronger liveness criterion than obstruction-freedom: any lock-free implementation is also obstruction-free, but not vice versa~\cite{herlihy2020art}. 
%\end{definition}

\paragraph{Memory overhead.} Recall that the memory overhead of a bounded queue implementation is the amount of memory that must be allocated on top of the fixed memory required for storing the queue elements.

%We note that obstruction-freedom is weaker than lock-freedom, thus, if we will prove something for obstruction-free implementations it will also hold for lock-free ones.
%\va{Quite a tricky statement which should be revisited?}

\subsection{Algorithm Restrictions}

\paragraph{Atomic operations.}
Recall the model supports read, write, and \texttt{Compare-and-Set} operations. For simplicity, we do not support \texttt{Fetch-and-Add}, but it does not break our proof when not manipulating memory locations where elements are stored. In particular, using \texttt{Fetch-and-Add} for updating counters is allowed.

\paragraph{Values and metadata.}
We restrict our attention to algorithms that assume a clear separation between \emph{value-locations}, used exclusively to store queue elements, and \emph{metadata-locations}, used to store everything else (e.g., counters or descriptors). Importantly, we impose no restrictions on the number of bits that can be stored in memory locations for both metadata and data, allowing metadata-locations to hold either value or metadata. Notably, we require that values and metadata be indistinguishable~--- metadata without context can be seen as a value. Making values and metadata distinguishable would require stealing a bit, automatically imposing $\Theta(C)$ memory overhead. Therefore, we find this restriction reasonable. Note that we allow to store $\perp$ (\texttt{null}) in both value- and metadata-locations.

% We highlight that differentiating locations to value- and metadata- ones, 

%except for values.
%
%The assumption is reasonable if we do not have any possibility to distinguish values from the metadata. 
%
%[[PK redundant
%State-of-the-art algorithms typically maintain a one-bit marker values specifying the type of a memory location (value or metadata) which either reduces the universe of values or incurs $\Omega(1)$ extra memory for value, thus, trivially resulting in $\Omega(C)$ overhead in total. This overhead is very high since the bounded queue is typically larger than the number of processes. 
%]]

%[[PK appears redundant
% The distinction between two types of memory locations can be slightly relaxed, for example, by allowing metadata locations to hold either value or metadata.
% Notably, a memory location (both value- and metadata- locations) can store at most one value, which matches real-world queue implementations.

% In this case, our bound would also be correct~--- we need precisely this number of locations that store values.
% Otherwise, we should have used additional bits to distinguish them, incurring $\Theta(C)$ extra memory.
% %]]

\paragraph{Value-independence.}
Furthermore, we assume that the algorithm is agnostic to \emph{values}:
intuitively, the values passed to the operation only reflect the value-locations in the memory and not the steps performed by the algorithm.
We call such algorithms \emph{value-independent}.
A value-independent algorithm must be \emph{non-creative}, \emph{enqueue-value-independent} and \emph{dequeue-value-independent}.
Intuitively, (1)~the implementation of an operation can only manipulate with values stored in value-locations or provided as an argument (if it is an  \textsf{enqueue}), (2)~it does not matter what is enqueued, $x$ or $y$---the only difference is that $y$'s should be stored instead of $x$'s in the memory, and (3)~if the only difference in the memory is $x$ and $y$ in one value-location, a \textsf{dequeue} operation must return $y$ instead of $x$.
We present the formal definitions below.
To the best of our knowledge, all the state-of-the-art queue implementations satisfy these assumptions.

\begin{definition}\label{def:non-creative}
A BQ algorithm $A$ is \emph{non-creative} if any value written into a value-location (by write or \texttt{CAS}) during an operation $op$ or returned by it was previously read by $op$~from a value-location or received as an argument (if $op$ is \emph{enqueue}).
\end{definition}

\begin{definition}\label{def:put-arg-ind}
Let $S$ be the memory state reached by an execution $E$ of a BQ algorithm $A$.
Suppose $a_1, \ldots, a_k$ and~$v$ are distinct values that do not appear in value-locations in $S$.
Let a process $p$ be \emph{idle in $E$}, i.e., it did not execute any operation (neither enqueue nor dequeue) in $E$.
Let $S_1$ be the memory state reached from $S$ when $p$ sequentially executes \texttt{enq($a_1$)}, \ldots, \texttt{enq($a_i$)}, \ldots, \texttt{enq($a_k$)} and $S_2$ be the memory state of $A$ reached from $S$ when $p$ sequentially executes \texttt{enq($a_1$)}, \ldots, \texttt{enq($v$)},\ldots, \texttt{enq($a_k$)} ($v$ is inserted instead of $a_i$). 
$A$ is \emph{enqueue-value-independent} if $S_1$ is identical to $S_2$, except that every value-location storing $a_i$ in $S_1$ stores $v$ in $S_2$.    
\end{definition}

%\vspace{-0.3cm}
\begin{definition}\label{def:extract-arg-ind}
Let $S$ be a memory state reached by some execution $E$ of a BQ algorithm $A$.
Suppose that there is a value $x$ that appears only in one value-location in $S$ and let process $p$ be idle in $E$.
Let process $p$ sequentially apply \texttt{dequeue()} operations as long as they are successful, i.e., until the queue becomes empty, and 
suppose that exactly one of these \texttt{dequeue()} operations $op$ returns $x$.
$A$ is \emph{dequeue-value-independent} if replacing $x$ with $y$ in $S$ and replaying \texttt{dequeue()} operations by process $p$, we obtain an execution of $A$ in which $op$ now returns $y$.
\end{definition}

\begin{definition}\label{def:arg-ind}
A BQ implementation $A$ is \emph{value-independent} if it is \emph{non-creative}, \emph{enqueue-value-independent}, and \emph{dequeue-value-independent}.
\end{definition}

%[[PK redundant
%Let $A$ be a value-independent BQ algorithm $A$ that uses $C + X$ value locations, where $X$ is the overhead.
%
%Our goal is to show that $X = \Omega(n)$ where $n$ is the number of processes. 
%]]
%

% We proceed by contradiction. 
% %
% Assuming that there is an obstruction-free linearizable BQ value-independent algorithm $A$ that employs too little memory (more precisely, $X<\lfloor \frac{T - 24}{24}\rfloor$), we establish a contradiction by constructing an execution of $A$ that cannot be linearized.

\subsection{Auxiliary Notions}

\begin{definition}
In the proof, we construct an execution where some processes are  \emph{paused} just before performing $CAS(l, x, y)$ operation. 
Such processes are called \emph{poised}, and the location $l$ is said to be \emph{covered} by that process.
\end{definition}

For simplicity, we treat $\textsf{l.write}(y)$ as $CAS(l, \ast, y)$ with  ``$\ast$'' matching any value. 

\begin{definition}
If a value $v$ was never used as an argument in the execution $E$, we say that $v$ is \emph{fresh} in $E$.

Let $E$ be a finite execution of $A$, and $p$ be an idle process in $E$.
A \emph{fill} procedure applied to $E$ (or just fill) consists of a process $p$ enqueueing $C$ different values in isolation.
We say that a fill procedure is \emph{successful} if all $C$ enqueues are successful and $C$ different value-locations store these values.

An \emph{empty} procedure (or just empty) consists of process $p$ executing $C$ dequeue operations in isolation.
We say that an empty procedure is \emph{successful} if all these $C$ dequeue operations are successful.

We say that a fill-empty-procedures pair is \emph{up-to-date} if both procedures are successful and dequeues from the empty procedure return the values enqueued by the fill.
\end{definition}

\begin{lemma}
    \label{lem:successful_fill_empty}
    Every finite execution has a \emph{solo} extension (where only a single process takes steps) that ends with an up-to-date fill-empty procedure. 
\end{lemma}

To prove Lemma \ref{lem:successful_fill_empty}, we need the following lemmas first.

\begin{definition}
    We say that the queue is \emph{logically empty} if it is in a state for which a new dequeue operation would return $\perp$ if executed in isolation.
\end{definition}

\begin{lemma}
\label{lem:restarts}
Let $S$ be a memory state produced by an execution $E$ in which the queue is logically empty queue and some of the processes are poised. 
If we extend $E$ with a \emph{fill} procedure followed by an \emph{empty} procedure and one of these procedures is not successful, then one of the operations executed by the posed processes ``took effect'', i.e., is linearized before any future operation.
\end{lemma}
\begin{proof}
We have a logically empty queue and we perform a \emph{fill} procedure.
Suppose this fill is not successful.
By definition, there can be two cases.
First, some enqueue operation is not successful then another enqueue operation succeded (exactly, $C$ enqueues succeed)~--- which can only be the one that was poised and, thus, it is linearized.
Second, all enqueue operations are successful, but not all $C$ new values are stored in the shared memory. Thus, there exists some poised dequeue operation that can return the non-stored value. But after the successful \emph{empty} procedure, it would mean that this poised operation is linearized.

After the \emph{fill} procedure, the logical state of the queue is full. Thus, $C$ dequeues should succeed. If the empty procedure is not successful, thus, some old poised dequeue takes place. This means that it was already linearized before that moment.
\end{proof}

\begin{lemma}
\label{lem:restarts2}
Consider some state of a system: a logically empty queue with poised operations. If we perform a successful \emph{fill} procedure followed by a successful \emph{empty} procedure and dequeue operations of the empty procedure return not the exactly same values that were enqueued by the fill procedure, then one of the poised operations ``took effect''.
\end{lemma}
\begin{proof}
Let us look at the value $v$ that was enqueued and was not dequeued by the empty procedure. In the end, the queue is empty; thus, the value $v$ was dequeued by some operation. Since it was not returned by any operation from the empty procedure, it should be some poised operation. Thus, after the empty procedure, this poised operation is linearized.
\end{proof}

By combining the last two lemmas, we construct the following result.

\begin{lemma}
\label{lem:restart}
Consider some state of a system: a logically empty queue with poised operations.
If a fill-empty-procedures pair applied to this state is not up-to-date, then one of the poised operations ``took effect'', i.e., is linearized before any future operations.
\end{lemma}

We are now ready to prove Lemma~\ref{lem:successful_fill_empty} and show that every execution can be extended by one working process with an up-to-date fill-empty procedure at the end. 

\begin{proof}[Proof of Lemma \ref{lem:successful_fill_empty}]
    In our proof, a dedicated process repeatedly performs alternating \emph{fill} and \emph{empty} procedures, 
    until the corresponding pair is up-to-date.
    By Lemma~\ref{lem:restart}, the number of such iterations cannot exceed $T$, as in our proof we catch at most $T-1$ concurrent operations
\end{proof}

%The proof of Lemma \ref{lem:successful_fill_empty} is rather technical and does not provide much insight into the subject. However, we still need to use it in our proofs, therefore, we provide a proof in Appendix~\ref{sec:lb_technic}.

%%%%%%%%%%%%%%%%%%%%%%%%%%%%%%%%%%%%%%%%%%%%%%%%%%

\subsection{The Lower Bound: Complete Proof}
\label{sec:app:lb}

We now present a complete proof of the memory overhead lower bound. Specifically, we prove Theorem~\ref{thm:lb} below. 

\begin{theorem}
 \label{thm:lb}
Any obstruction-free linearizable value-in\-de\-pen\-dent implementation of Bounded Queue with capacity $C$ uses $C + \lfloor \frac{T - 24}{24} \rfloor$ value-locations, assuming that there exists at least $\Theta(T \cdot C)$ different values and $\frac{T}{2} < C$.
%\ap{It was $n < \frac{C}{2}$ before...}
\end{theorem}

Recall that the key idea is to construct a non-linearizable execution. 
Before diving deeper into details, we describe the execution we target to construct and prove that it is non-linearizable.

\begin{lemma}
    \label{lem:universal}
    Let $k = 2X + 3$. 
    Let $E$ be an execution that results in a state that has $k$ value-locations $l_1, \ldots, l_k$, which store distinct values $x_1, \ldots, x_k$, and $x_i$ can be replaced with $y_i$ in each value-location $l_i$ by poised \texttt{CAS}-s.
    %~--- we either have $CAS(l_i, x_i, y_i)$ or a pair $CAS(l_i, x_i, \perp)$ with $CAS(l_i, \perp, y_i)$. 
    % 
    Also, all $x_i$ are distinct, and for all $i$, $x_i \neq y_i$, and $x_i, y_i$ are non-$\bot$.
    
    Then we can extend $E$ to a non-linearizable execution.
\end{lemma}
\begin{proof}
    First, we extend $E$ using one process so it ends with up-to-date fill-empty pair (by Lemma \ref{lem:successful_fill_empty}, this is possible). 
    By the definition of a successful fill, there are $C$ distinct values stored in value-locations after it, so, since we have $C + X$ value-locations in total, there are at least $C - X$ of these new values that are not duplicated in value-locations, i.e. each of them appears in exactly one value-location.
    As $k = 2X + 3$, after the fill, there are at least $3$ locations among $l_1, \ldots, l_k$ such that a value of some enqueue from fill is stored only in that location. 
    Thus, there is $l_i$ that stores a value of some $enq$ from the fill procedure that is not a first $enq$, nor the last one, and this value is stored only in $l_i$. Using enqueue-value-independence we change the argument of the corresponding enqueue to $x_i$, reaching the state where after the fill $l_i$ is the only location storing $x_i$ and $x_i$ is returned in the following empty procedure. 
    		
    Therefore, if we replay this updated fill procedure we could perform poised \texttt{CAS}-s substituting $x_i$ with $y_i$ before the empty procedure. Thus, if at this point we replay the empty procedure, the dequeue that used to return $x_i$ now returns $y_i$ by dequeue-value independence.
    		
    By that, the last fill-empty-procedures pair have operations
    \begin{center}
    \begin{tabular}{|c|c|c|c|c|}
    \hline
    $enq(v_1)$ & $\cdots$  & $enq(x_i)$ & $\cdots$ & $enq(v_C)$\\
    \hline
    $deq \rightarrow v_1$ & $\cdots$ & $deq \rightarrow y_i$ & $\cdots$ & $deq \rightarrow v_C$\\
    \hline
    \end{tabular}  
    \vspace{0.2cm}
    %Figure 4.
    \end{center}
    		
    We call a pair of type $\langle enqueue(v),\ dequeue \rightarrow v\rangle$ a \emph{matching} pair and it is \emph{true-matching} if in any given linearization, these enqueue and dequeue indeed match each other.\pk{Linearization of what? Matching where in which execution?}
    		
    We claim that $\langle enq(v_1), deq \rightarrow v_1\rangle$ and $\langle enq(v_C), deq()\rightarrow v_C\rangle$ are true-matching pairs in the constructed execution. 
    		
    To see this, note that for a matching pair to be not true-matching there should be another enqueue whose value is dequeued. But $v_1$ and $v_C$ are fresh values since the only non-fresh value among enqueues is $x_i$ and by the construction, it is not the last, nor the first enqueue. Thus, there can not be other poised $enq(v_1)$ or $enq(v_C)$. 
    		
    The question now is where to place $enq(y_i)$ to find a true-match ${deq \rightarrow y_i}$. It cannot be placed before $enq(v_1)$ or after $enq(v_C)$ since the pairs for $v_1$ and $v_C$ are true-matching and they cannot transpose due to FIFO property of the queue. So, it can only be placed between $enq(v_1)$ and $enq(v_C)$. 
    		
    But then we have $C+1$ successful enqueues in a row, so there must be $deq \rightarrow v_1$ before $enq(v_C)$ which means that the pair $\langle enq(v_1), deq \rightarrow v_1 \rangle$ from the table is not a true-match, but we have proven it already. 
    
    Thus, $enq(x_i)$ corresponds to $deq \rightarrow y_i$ which makes the execution not linearizable.
\end{proof}

% The following lemma is proved via a simple pigeon-hole argument:  
% \begin{lemma}
% \label{lem:unique}
% For the state of the queue with $C$ distinct values in value-locations, among each $2X + t \leq C + X$ locations there are at least $t$ in which unique values are stored, i.e., these values are stored only in these $t$ locations.
% \end{lemma}
% \begin{proof}
% Consider set of $2X + t$ value-locations $U$. Let $V$, $|V| = C$, be the set of value-locations containing $C$ distinct values. The intersection of $U$ and $V$ has at least $X + t$ elements since there are only $C + X$ locations in total. Consider a location $l \in U \cap V$. If value stored in $l$ is not unique, its duplicate might only be stored in some location $l' \in U \setminus V$, since all values in $V$ are distinct. Moreover, one element in $U \setminus V$ can violate uniqueness of at most one element in $U \cap V$ since all elements in $U \cap V$ are distinct since it is a subset of $V$. Thus, there are at least $t$ unique elements in $U \cap V$, since $|U \cap V| \geq X + t$ and $|U \setminus V| = |U \setminus (U \cap V)| \leq X$.
% \end{proof}

We are finally ready to prove the main Theorem~\ref{thm:lb} and show that any obstruction-free linearizable value-in\-de\-pen\-dent implementation of Bounded Queue with capacity $C$ uses $C + \lfloor \frac{T - 24}{24} \rfloor$ value-locations, assuming that there exists at least $\Theta(T \cdot C)$ different values and $\frac{T}{2} < C$.

\begin{proof}[Proof of Theorem \ref{thm:lb}]
The proof proceeds by contradiction. 
Namely, we pick the number of additional value-locations $X$ small enough, though still linear in $T$, and show that for each implementation that uses at most $X$ additional value-locations, we can build a non-linearizable execution. 
We prove that it is sufficient to pick $X = \lfloor \frac{T - 24}{24} \rfloor$.
	
Our procedure to create non-linearizable execution is partitioned into four steps. 
During the first step, the queue is filled with fresh values, and $\frac{T}{2}$ processes are poised right before the CAS operation. 
In the second step, we just call empty and fill procedures in order to obtain fresh values in memory. The third and fourth steps are the most complicated. Depending on the BQ algorithm, we get three cases depending on poised operations. 
The goal of each case is to come up with an execution that brings the queue to the state required by Lemma~\ref{lem:universal}. 
Finally, using this lemma we build a non-linearizable execution.

Now, we dive into the details. Our construction algorithm works in four steps.

\paragraph{First step.} We try to extend an execution $\frac{T}{2}$ times with an up-to-date fill-empty pair using Lemma~\ref{lem:successful_fill_empty}. For each pair, we require a new \emph{idle} process, i.e., a process that has not started any operation. Though, for each try, we don't let a process finish an extension, rather we stop this process right before it is about to CAS at a not yet covered value-location from $\bot$. We say that we \emph{catch} the process when we stop it right before performing a CAS. When we catch the process, we restart with another process using Lemma~\ref{lem:successful_fill_empty}. We claim that using this procedure we can catch $\frac{T}{2}$ processes.

To see that, assume by the contrary that we do not catch $\frac{T}{2}$ processes. Thus, one of them managed to perform a successful fill. By the definition of successful fill, there must be $C$ fresh values stored in the value-locations after it, so a process made successful attempts to CAS $C$ value-locations and was not poised. This implies that all those locations are covered since all locations initially store $\bot$ and hence cannot be modified without being previously covered. But since one process covers only one location and $C > \frac{T}{2}$, $C$ locations cannot be covered, a contradiction. 

\paragraph{Second step.}
We extend an execution using a new idle process to end up with a successful fill. It is indeed possible by the trivial corollary of the Lemma \ref{lem:successful_fill_empty}: if we can ensure successful fill and empty, we can just stop after the fill. 

\paragraph{Third step.}
We try to extend an execution $\frac{T}{2}$ times with an up-to-date fill-empty pair using Lemma~\ref{lem:successful_fill_empty}. For each pair, we require a new idle process. Though, for each try, we do not let a process finish an extension, rather we stop this process right before it is about to CAS on a location that satisfies the following \emph{catch criteria}:
\begin{enumerate}[topsep=0pt]
    \item The corresponding value-location was not yet covered by some CAS at Step 3;
    \item This value-location is already covered by some CAS at Step 1;
    \item The poised CAS must be from a non-bottom value different from the previous \texttt{CAS}-s poised during Step~3.
\end{enumerate}

\paragraph{Fourth step.}
This step is to account for the case when there are few processes caught at Step 3.

If there are $\frac{T}{4} - 3X$ idle processes at this point, we consecutively extend the execution with a successful fill using each of them. It is indeed possible by the trivial corollary of the Lemma \ref{lem:successful_fill_empty}. We catch a process if it satisfies the following \emph{catch criteria}:
\begin{enumerate}
    \item The corresponding value-location is not yet covered by CAS at Step 4;
    \item This value-location is already covered by CAS at Step~1;
    \item The poised CAS must be from a non-bottom value (to possibly a bottom one) different from the previous \texttt{CAS}-s poised during Step 4.
\end{enumerate}
		
Now, we explain how our poised operations provide a state desired by Lemma~\ref{lem:universal}. This part becomes a little bit technical~--- we get \emph{three} cases depending on how many \texttt{CAS}-s we poised during Step 3:
\begin{itemize}
    \item Suppose that we poised less than $\frac{T}{4}$ processes. 
			
    After Step 2, the successful fill is performed, so $C$ fresh distinct values are stored in $C$ different locations. We have $\frac{T}{2}$ value-locations covered by \texttt{CAS}-s at Step 1, so, since we have $C + X$ value-locations in total, there are $\frac{T}{2} - 2X$ locations among covered that store unique values. Let us call the set of these locations that are not modified during Step 3 as $L'$. We claim that $|L'| > \frac{T}{4} - 2X$.
			
    To see this, assume the opposite, namely that $L'$ has at most $\frac{T}{4} - 2X$ value-locations. Then at least $(\frac{T}{2} - 2X) - |L'| = \frac{T}{4}$ of our chosen value-locations were modified. It means that during our procedure in Step 3, we made a CAS on these locations but did not catch it. This could happen only if one of the catch criteria of Step 3 was not satisfied.
    However, criteria (2) and (3) are always satisfied due to the way we chose locations: the chosen value-locations are already covered by some CAS from Step 1, and they stored non-bottom value (they store some of $C$ values from Step 2). Thus, the only criterion that could not be satisfied is (1), i.e., it is already covered by a CAS. But, we poised less than $\frac{T}{4}$ \texttt{CAS}-s in the main condition; thus, (1) is satisfied for some location~--- we could have poised at least one more CAS. Thus, $|L'|$ is greater than $\frac{T}{4} - 2X$.
			
    We claim that during Step 4, we catch $\frac{T}{4} - 3X$ processes. Otherwise, one of them was able to complete a successful fill, which implies that $C$ fresh values are stored, so at most $X$ locations are left unmodified. Thus, at least $|L'| - X = \frac{T}{4} - 3X$ locations in $L'$ must be overwritten~-- there will be attempts to successfully CAS and hence cover them. Let $L$ be the set of value-locations covered by poised operations from Step 4.		
    Finally, we apply Lemma~\ref{lem:universal} with $l_i$ being locations from $L$.

    We now need to show that the Lemma~\ref{lem:universal} requirements are satisfied. By the construction of $L$, all $l_i \in L$ are indeed distinct. For each $1 \leq i \leq |L|$ there is a poised CAS of the form $CAS(l_i, \bot, y_i)$ and a CAS of the form $CAS(l_i, x_i, mb_i)$. If $mb_i \neq \bot$, then we can change the value in $l_i$ from $x_i$ to $mb_i$. Otherwise, we can perform both \texttt{CAS}-s changing the value in $l_i$ from $x_i$ to $y_i$. We also need to show that all $x_i$ are distinct. It follows from the fact that all values in $L'$ are. Finally, $|L| = \frac{T}{4} - 3X \geq 2X + 3$, thus, we can choose $2X+3$ locations from $L$ and apply Lemma~\ref{lem:universal}.
			
    \item Suppose that in Step 3 at least $\frac{T}{4}$ processes are poised while more than $\frac{T}{8}$ of them CAS to non-bottom.
   
    Let $L$ be the set of value-locations covered by poised \texttt{CAS}-s updating to non-bottom.
			 
    Now, we can apply Lemma \ref{lem:universal} with $l$-s being locations from $L$.

    Let us check Lemma's requirements are satisfied. Locations in $L$ are distinct by the catch criteria. We can indeed substitute $x_i$ with $y_i$ in each $l_i$ applying the corresponding CAS. $x_i$ are non-bottom because of the catch criteria and $y_i$ are non-bottom because of the current case assumption. All $x_i$ are distinct because of the catch criteria. Finally, $|L| \geq \frac{T}{8} \geq 2X + 3$, thus, we can choose $2X+3$ locations from $L$ and apply Lemma~\ref{lem:universal}.
			
    \item Suppose that in Step 3 at least $\frac{T}{4}$ processes are poised and at least $\frac{T}{8}$ of them CAS to bottom.
			
    Here, we have at least $\frac{T}{8}$ pairs of the type $CAS(l_i, \bot, y_i)$ and $CAS(l_i, x_i, \bot)$.
			
    Note that $x$-s are those values that were present in the memory after Step 2, whereas $y$-s are those introduced during Step 1. Recall that by non-creative property $y$-s are all arguments of enqueue operations.
			
    Filling the queue with fresh values in Step 2 guarantees that at most $X$ values out of those $y$-s are present at the start of Step 3.
    Since all $x$-s are pairwise distinct by the catch criteria of Step 3, there can be no more than $X$ values from $y$-s among them. We no longer consider CAS pairs where CAS is done from $x$ which is equal to some $y_i$. However, there are at most $X$ of the non-satisfying pairs and we still have at least $\frac{T}{8} - X$ good $CAS(l_i, \bot, y_i), CAS(l_i, x_i, \bot)$ pairs.
    Let $L = \{l_i\}$ be the set of value-locations covered by these pairs. 
			
    Finally, we can apply Lemma~\ref{lem:universal} with $l_i$ being elements of $L$.

    Requirements of the Lemma are satisfied. Indeed, $l$-s are distinct and $x$-s are distinct by the catch criteria of Step 3. For each $l_i$ we can substitute $x_i$ with $y_i$ performing a pair of \texttt{CAS}-s. Moreover, $x$-s are distinct from $y$-s, since $y$-s are values from Step 1 and by the construction of $L$ there are no values from Step 1. Finally, $|L| = \frac{T}{8} - X \geq 2X + 3$,  thus, we can choose $2X+3$ locations from $L$ and apply Lemma~\ref{lem:universal}.
    \end{itemize} 
\end{proof}

\paragraph{System-wide overhead.}
Intuitively, one may share the overhead between multiple queues. Specifically, when using $k$ obstruction-free bounded queues of capacity $C$, the total memory overhead might remain $O(T)$ (not depending on $k$). However, our primary result is the impossibility of achieving constant overhead in a single bounded queue, significantly affecting the algorithm design and its trade-offs.

% \begin{remark}
% It is possible to restrict the number of different values from $\Theta(T \cdot C)$ to $\Theta(C)$.
% For that, we can re-use the old values that already do not appear in the memory as the fresh ones.
% \end{remark}

\subsection{The Upper Bound}
Finally, we show that our lower bound is tight by presenting an algorithm that asymptotically matches the lower bound and uses $O(T)$ extra memory. Please note that our goal is to show the upper bound of the memory overhead rather than provide a practical and efficient implementation.  

\paragraph{Algorithm overview.}
The key idea is to use descriptors for \texttt{enqueue(..)} operations,
with the possibility of reusing them~\cite{tbrown_recyclable_desc}, to match the desired memory overhead.
% For that our algorithm uses descriptors stored in pre-allocated metadata locations (each descriptor takes $\Theta(1)$ memory). 
%
% To allow their re-use we \emph{recycle} these descriptors by using the approach from.
%
% For that, we need to allocate no more than $2 \cdot T$ descriptors in total.
%\pk{Do not see it in the code: line~\ref{line:putop:opslot:choose} suggests that only $n$ slots are used in \texttt{ops}.}
%
In addition, we use an $T$-size ``announcement'' array that stores references to descriptors of ``in-progress'' \texttt{enqueue(..)} invocations. 
Intuitively, a descriptor in this array declares an \emph{intention} to perform an enqueue operation on some chosen cell.
This enqueue succeeds if the \texttt{enqueues} counter has not been changed (which is similar to our \texttt{DCSS}-based solution above).
The descriptor stores a target cell to write and an operation status in \texttt{successful} field.
We say that a descriptor/operation/thread ``covers'' a cell if it intends to put an element there.

Thus, in total, we have $O(T)$ memory overhead that asymptotically matches the lower bound shown in the previous subsection. 

Now, we explain how descriptors and the announcement array help us achieve a correct algorithm.
As in the previous algorithms, a \texttt{dequeue()} (resp., \texttt{enqueue(..)}) operation starts with taking a snapshot of the counters in lines and checks if the queue is \emph{empty} (resp., \emph{full}). 

The implementation of \texttt{dequeue()} slightly differs from those in the algorithms from Section~\ref{sec:special-cases}. 
To read an element, it goes through the announcement array to check for the ongoing enqueue operations on the target cell, which is calculated from the counters.
If one exists, our dequeue takes the argument of that ongoing operation.
Otherwise, it takes a value from the target cell.
Finally, it tries to increment the \texttt{dequeues} counter and returns the found element if the corresponding \texttt{CAS} succeeds. 

As for \texttt{enqueue(..)}, it creates a special descriptor, announces it in our announcement array, and then tries to atomically apply the operation. 
Then, the operation increments the \texttt{enqueues} counter, possibly helping some concurrent operation. 
If the descriptor is successfully applied, the operation completes.
Otherwise, the operation is restarted.

There can be only one descriptor in the successful state that can ``cover'' a given cell in the elements array, no other successful descriptor can point to the same cell. 
Thus, only the thread that started covering the cell is eligible to update it.
When the thread with exclusive access for modifications finishes the operation, it ``releases'' the cell, so, another operation is able to ``cover'' it. 
However, when an enqueue operation finds a cell to be covered, it replaces the old descriptor (related to the thread that covers the cell now) with a new one and completes. 
The thread with the exclusive access to the cell helps this enqueue operation to put the element. 

Summing everything up, our algorithm is lock-free (Appendix~\ref{sec:algo_lf}) and it uses just $O(T)$ additional memory to $C$ values in the bounded queue.
Please note that we do not provide any guarantees about the execution time of operations on our BQ.
One should expect the queue to answer the queries in $O(1)$, however, in our algorithm each operation should traverse through the whole announcement array leading to $\Theta(T)$ time.
This leaves us with an open question whether it is possible to implement a memory-optimal queue that serves requests in $O(1)$ time.
%%%%%%%%%%%%%%%%%%%%%%%%%%%%%%%%%%%%%%%%%%%%%%%%

\paragraph{Implementation details.} We discuss the implementation details and prove the algorithm's correctness in Appendix~\ref{sec:upper-bound}.

\section{Related Work}\label{sec:related}
%\vspace{-0.2em}
Memory efficiency has always been one of the central concerns in concurrent computing. 
Many theoretical bounds have been established on memory requirements of various concurrent abstractions, such as mutual exclusion~\cite{BL93}, perturbable objects~\cite{JTT00}, or consensus~\cite{Zhu16}.
However, it appears that minimizing memory overhead in dynamic concurrent data structures has not been in the highlight until recently.     
A standard way to implement a lock-free bounded queue, the major running example of this paper, is to use descriptors \cite{valois1994implementing,pirkelbauer2016portable} or additional meta-information per each element \cite{tsigas2001simple,vyukov,shafiei2009non,feldman2015wait}. 
The overhead of resulting solutions is proportional to the queue size: a descriptor contains an additional $\Omega(1)$ data to distinguish it from a value, while an additional meta-information is $\Omega(1)$ memory appended to the value by the definition. 

The fastest queues, however, store multiple elements in nodes so the memory overhead per element is relatively small from the practical point of view~\cite{morrison2013fast, yang2016wait, LPRQ}, and we consider them as memory-friendly, though not memory-optimal.

\noindent
A notable exception is the work by Tsigas et al.~\cite{tsigas2001simple} that tries to answer our question: whether there exists a lock-free concurrent bounded queue with $O(1)$ additional memory.
The solution proposed in~\cite{tsigas2001simple} is still a subject to the ABA problem even when all the elements are different: it uses only two null-values, and if one process becomes asleep for two ``rounds'' (i.e., the pointers for \texttt{enqueue(..)} and \texttt{dequeue()} has made two traversals through all the elements), waking up it can incorrectly place the element into the queue. 
Besides resolving the issue, our algorithm, under the same assumptions, is shorter and easier to understand (see Subsection~\ref{subsec:ppopp}).

The tightest algorithm we found is the recent work by Nikolaev~\cite{nikolaev2019scalable}, that proposes a lock-free bounded queue with capacity $C$ implemented on top of an array with $2C$ memory cells. While the algorithm manipulates the counters via \texttt{Fetch-And-Add}, it still requires descriptors, one per each ongoing operation. This leads to the additional overhead linear in $T$. Thus, the total memory overhead is $\Omega(C + T)$, while the lower bound is $\Theta(T)$ and does not depend on the queue capacity.
%         

%%%%%%%%%%%%%%%%%%%%%%%%%%%%%%%%%%%%%%%%%%%%%%%%%%%

%\vspace{-0.3cm}

\section{Discussion} \label{sec:disc}
%\vspace{-0.2em}
In this paper, we proved that any non-blocking implementation of a bounded queue incurs $\Omega(T)$ memory overhead, while showing that the bound is asymptotically tight with a matching algorithm.
In Section~\ref{sec:special-cases}, we also presented a series of algorithms with \emph{constant} memory overhead that work under several practical restrictions on the system or/and the application, along with a simple \texttt{DCSS}-based algorithm that matches the lower bound but requires an ability to store values and descriptors in the same array slot. 

%[[PK Tautology?
%Furthermore, we find it important to adjust the lower bound and provide the tight solutions under a set of practical relaxation. 
We find the following relaxations important to consider in future work: (1)~the single-producer and single-consumer application restrictions, (2)~the ability to store descriptors in value-locations (free in JVM and Go) or ``steal'' a couple of bits from addresses (free in C++) for containers of references, and (3)~relaxing the object semantics, e.g., to probabilistic guarantees.
Each of these cases corresponds to a popular class of applications, and determining the optimal memory overhead in these scenarios is appealing in practice.  
%[[PK pretentious?
%We believe that solving these questions makes an outstanding improvement for the academia and the industry .
%]]

%Talking about the industry, the current trend to make 

% A popular trend to make concurrent containers scalable is to use \texttt{Fetch-And-Add} (\texttt{FAA}) on the contended path and \texttt{CAS} for remaining  synchronization. 
% As for queues, the fastest algorithms we are aware of use \texttt{FAA} to increment the counters, so at most one \texttt{enqueue(..)} and \texttt{dequeue()} manipulate with the same value-locations~\cite{morrison2013fast, yang2016wait,LPRQ}.
% Following this pattern, it would be interesting to restrict the number of parallel accesses to value- or, even, metadata-locations, adjusting the lower bound and providing a matching algorithm.

%[[PK not very informative
%Since the presented in the paper algorithm ensures only lock-freedom, the existence of a wait-free one is questionable, especially taking into account that wait-free implementations are typically much more complicated.
%]]
%

It would also be interesting to consider the problem of memory overhead incurred by \emph{wait-free} implementations since such solutions are typically more complicated.  

Finally, in this paper, we only focused on bounded queues, as they allow a natural definition of memory overhead.
One may try to extend the notion to the case of \emph{unbounded} queues by considering the \emph{ratio} between the amount of memory allocated for the metadata with respect to the memory allocated for data elements. 
Defining the bounds on this ratio for various data structures, such as pools, stacks, etc., remains an intriguing open question.

\begin{acks}
The authors would like to thank the anonymous referees and the shepherd for their valuable comments and helpful suggestions.
Petr Kuznetsov was supported by TrustShare Innovation Chair (Mazars \& CDD).
\end{acks}

\bibliographystyle{ACM-Reference-Format}
\bibliography{references}

\appendix

%%%%%%%%%%%%%%%%%%%%%%%%%%%%%%%%%%%%%%%%%%%%%%%%%%%

\clearpage
\section{Memory-Optimal Bounded Queue}
\label{sec:upper-bound}

In this section we present an algorithm (Figure~\ref{lst:complicated}) that exhibits  $O(T)$ memory overhead only using read, write, and CAS primitives.
In Section~\ref{sec:app:lb}, we showed that the algorithm is (asymptotically) memory-optimal.

To address the ABA problem, in our algorithm, \texttt{enqueue(..)} operations use descriptors stored in pre-allocated metadata locations (each descriptor takes $\Theta(1)$ memory). 
To enable \emph{recycling}~\cite{tbrown_recyclable_desc}, we allocate $2 \cdot T$ descriptors.
For simplicity, we omit that from the pseudo-code: but this reclamation procedure indeed happens.
%\pk{Do not see it in the code: line~\ref{line:putop:opslot:choose} suggests that only $n$ slots are used in \texttt{ops}.}
%
In addition, an $T$-size ``announcement'' array is used to store references to the descriptors.
Altogether, this gives $O(T)$ memory overhead. 

As in the previous algorithms, a  \texttt{dequeue()} (resp., \texttt{enqueue(..)}) operation starts with taking a snapshot of the counters in lines~\ref{line:hard:deq_dc_start}--\ref{line:hard:deq_dc_end} (resp., lines~\ref{line:hard:enq_dc_start}--\ref{line:hard:enq_dc_end})), and checks if the queue is \emph{empty} in line~\ref{line:hard:deq_emptiness} (resp., \emph{full} in line~\ref{line:hard:enq_fullness}). 

\paragraph{Overview.}
The implementation of \texttt{dequeue()} slightly differs from those in the algorithms from Section~\ref{sec:special-cases} above. 
To read an element to be retrieved during the snapshot (line~\ref{line:dequeue:readElem}), it uses a special \texttt{readElem} function. 
Then it tries to increment the \texttt{dequeues} counter and returns the element if the corresponding \texttt{CAS} succeeds (line~\ref{line:dequeue:inc}). 

As for \texttt{enqueue(..)}, it creates a special \texttt{EnqOp} descriptor  (declared in lines~\ref{line:enqop:start}--\ref{line:enqop:end}) and then tries to atomically apply the operation (line~\ref{line:enq_applydesc}). 
Then the operation increments the \texttt{enqueues} counter, possibly helping a concurrent operation (line~\ref{line:enqueue:inc}). 
If the descriptor is successfully applied, the operation completes.
Otherwise, the operation is restarted.
%[[PK why otherwise, appears unconditional
%Before that it tries to increment the counter, possibly helping  \emph{helps} another operation (line~\ref{line:enqueue:inc}) with the counter increment and restarts.

The algorithm maintains an array \texttt{ops} of \texttt{EnqOp} descriptors (line~\ref{line:enqops}), which specifies ``in-progress'' \texttt{enqueue(..)} invocations. 
Intuitively, an \texttt{EnqOp} descriptor declares an \emph{intention} to perform an enqueue operation that succeeds if the \texttt{enqueues} counter has not been changed (which is similar to our \texttt{DCSS}-based solution above).
The descriptor stores an operation status in \texttt{successful} field (line~\ref{line:enqop:successful}).
We say that a descriptor/operation/thread ``covers'' a cell if it has an intention to put an element there.
%[[PK not needed?
%, which is undefined on creation. 
%]]
%

There can be only one descriptor in the successful state that can ``cover'' a given cell in the elements array \texttt{a}, no other successful descriptor can point to the same cell (see \texttt{putOp} function in lines~\ref{line:putop:start}--\ref{line:putop:end}). 
Thus, only the thread that started covering the cell is eligible to update it.
%and %that by that 
%there is no conflict on updates here. 
%
When the thread with an exclusive access for modifications finishes the operation, it ``releases'' the cell (see function \texttt{completeOp} in lines~\ref{line:completeOp:start}--\ref{line:completeOp:end}), so, another operation is able to \emph{cover} it. 
However, when an enqueue operation finds a cell to be covered, it replaces the old descriptor (related to the thread that covers the cell now) with a new one and completes (lines~\ref{line:apply:op_success_true}-\ref{line:apply:op_success_false}). 
The thread with the exclusive access to the cell helps this enqueue operation to put the element into \texttt{a}. %\pk{Did not get this cover/uncover thing ((}

\begin{figureAsListingWide}
\begin{minipage}[t]{0.47\textwidth}
\begin{lstlisting}[basicstyle=\footnotesize\selectfont\ttfamily]
// Descriptor for enqueues #\label{line:enqop:start}#
class EnqOp<T>(enqueues: Long,
               element: Type) { 
 val e = enqueues // `enqueues` value
 val x = element  // the inserting element
 val i = e % C // cell index in `a`
 // Op status: true, false, or #\color{Mahogany}$\perp$#;
 // we consider #\color{Mahogany}$\perp$# as the second 
 // 'false' in logical expressions.
 var successful: Bool? = #$\perp$# #\label{line:enqop:successful}#

 fun tryPut(): // performs the logical put
   // Is there an operation which
   // already covers cell 'i'?
   (op, opSlot) := findOp(i) #\label{line:tryput:checkCovered:start}#
   if op != #$\perp$# && op != this:
   #\indentrule#  CAS(&successful, #$\perp$#, false) #\label{line:tryput:checkCovered:end}#
   // Has 'enqueues' been changed?
   eValid := e == enqueues #\label{line:tryput:checkEnqChanged:start}#
   CAS(&successful, #$\perp$#, eValid) #\label{line:tryput:checkEnqChanged:end}#
} #\label{line:enqop:end}#

// Currently running enqueues
var ops: EnqOp?[] = new EnqOp?[T] #\label{line:enqops}#
// The next operation to be applied
var activeOp: EnqOp? = #$\perp$# #\label{line:activeop}#

func deq() Type? = while (true):
  #\indentrule#  d := dequeues; e := enqueues #\label{line:hard:deq_dc_start} \label{line:deq:read:e}#
  #\indentrule#  x := readElem(d % C) #\label{line:dequeue:readElem}#
  #\indentrule#  if d != dequeues: continue #\label{line:hard:deq_dc_end}#
  #\indentrule#  if e == d: return #$\perp$# #\label{line:hard:deq_emptiness}#
  #\indentrule#  if CAS(&dequeues, d, d + 1): return x #\label{line:dequeue:inc}#  

func enq(x: Type) Bool = while (true):
  #\indentrule#  e := enqueues; d := dequeues #\label{line:hard:enq_dc_start} \label{line:enq:read:d}#
  #\indentrule#  if e != enqueues: continue #\label{line:hard:enq_dc_end}#
  #\indentrule#  if e == d + C: return false #\label{line:hard:enq_fullness}#
  #\indentrule#  op := new EnqOp(e, x); apply(op) #\label{line:enqop:create}#  #\label{line:enq_createdesc} \label{line:enq_applydesc}#
  #\indentrule#  CAS(&enqueues, e, e + 1) #\label{line:enqueue:inc}#
  #\indentrule#  if op.successful: return #\label{line:enqueue:successful:check}#

// Puts 'op' into 'ops' and returns 
// its location, or #\color{Mahogany}-1# on failure.
func putOp(op: EnqOp) Int: #\label{line:putop:start}#
  for j in 1..#$\infty$#: #\label{line:putop:occupy:start}#
  #\indentrule#  opSlot := j % T 
  #\label{line:putop:opslot:choose}#
  #\indentrule#  if !CAS(&ops[opSlot], #$\perp$#, op): #\label{line:putop:put}#
  #\indentrule#  #\indentrule#  continue // occupied #\label{line:putop:occupy:end}#
  #\indentrule#  startPutOp(op)
  #\indentrule#  op.tryPut() // logical addition #\label{line:putOp:tryPut}#
  #\indentrule#  // Finished, free `activeOp`.
  #\indentrule#  CAS(&activeOp, op, #$\perp$#) #\label{line:putop:clean_activeop}#
  #\indentrule#  if !op.successful:
  #\indentrule#  #\indentrule#  ops[opSlot] = #$\perp$# // clean the slot #\label{line:putop:clean}#
  #\indentrule#  #\indentrule#  return -1
  #\indentrule#  return opSlot #\label{line:putop:end}#
\end{lstlisting}
\end{minipage}
\hfill
\begin{minipage}[t]{0.47\textwidth}
\begin{lstlisting}[firstnumber=59, basicstyle=\footnotesize\selectfont\ttfamily]
// Starts the 'op' addition.
func startPutOp(op: EnqOp) = while (true): #\label{line:startPutOp:start}#
  #\indentrule#  cur := activeOp #\label{line:startPutOp:helpStart}#
  #\indentrule#  if cur != #$\perp$#: // need to help
  #\indentrule#  #\indentrule#  cur.tryPut() #\label{line:startPutOp:tryPut}#
  #\indentrule#  #\indentrule#  CAS(&activeOp, cur, #$\perp$#) #\label{line:startPutOp:helpEnd}#
  #\indentrule#  if CAS(&activeOp, #$\perp$#, op): return #\label{line:startPutOp:set}# #\label{line:startPutOp:end}#

// Only the thread that covers the
// cell is eligible to invoke this.
func completeOp(opSlot: Int) = while (true): #\label{line:completeOp:start}#
  #\indentrule#  op := readOp(opSlot) // not #\color{Mahogany}$\perp$# #\label{line:completeop:readop}# 
  #\indentrule#  a[op.i] = op.x #\label{line:completeOp:dump}#
  #\indentrule#  CAS(&enqueues, op.e, op.e + 1) #\label{line:enqueue:inc:1}#
  #\indentrule#  if CAS(&ops[opSlot], op, #$\perp$#): return #\label{line:completeOp:end}# #\label{line:com}#

// Tries to apply 'op'.  
func apply(op: EnqOp): #\label{line:apply:start}#
  (cur, opSlot) := findOp(op.i) #\label{line:apply:findOp}#
  // Try to cover the cell by 'op'.
  if cur == #$\perp$#: #\label{line:apply:checkCovered}#
  #\indentrule# opSlot = putOp(op) #\label{line:apply:tryPutOp}#
  #\indentrule# // Complete `op` if the cell is covered 
  #\indentrule# if opSlot != -1: completeOp(opSlot)#\label{line:apply:completeOp}#
  #\indentrule# return
  // 'cur' already covers the cell.
  if cur.e >= e: // is 'op' outdated? #\label{line:apply:checkOutdated:start}#
  #\indentrule#  op.successful = false
  #\indentrule#  return #\label{line:apply:checkOutdated:end}#
  // Try to replace 'cur' with 'op'.
  op.successful = true #\label{line:apply:op_success_true}#
  if CAS(&ops[opSlot], cur, op): return#\label{line:apply:replace}#
  // The replacement failed.
  op.successful = false #\label{line:apply:op_success_false}# #\label{line:apply:end}#

// Looks for 'EnqOp' that covers cell #\label{line:readElem:startComment}#
// 'i', returns 'a[i]' if not found.
func readElem(i: Int) Type: #\label{line:readElem:start}#
  (op, _) := findOp(i)
  if op != #$\perp$#: return op.elem
  return a[i] // EnqOp is not found #\label{line:readElem:end}#

// Returns the operation located at #\label{line:readOp:startComment}#
// 'ops[opSlot]' if it is successful.
func readOp(opSlot: Int) EnqOp?: #\label{line:readOp:start}#
  op := ops[opSlot]
  if op != #$\perp$# && op.successful: return op
  return #$\perp$# #\label{line:readOp:end}#

// Returns a successful operation that #\label{line:findOp:startComment}#
// covers cell 'i', #\color{Mahogany}$\perp$# otherwise.
func findOp(i: Int) (EnqOp?, Int): #\label{line:findOp:start}#
  for opSlot in 0..n-1 {
  #\indentrule#  op := readOp(opSlot)
  #\indentrule#  if op != #$\perp$# && op.i == i: 
  #\indentrule#  #\indentrule#  return (op.value, opSlot)
  return #$\perp$# // not found #\label{line:findOp:end}#
\end{lstlisting}
\end{minipage}
\vspace{1em}
\caption{Bounded queue algorithm with $\Theta(T)$ overhead.
%[[PK unnecessary?
%, which avoids the ABA problem by ensuring exclusive write access to each array cell via re-usable \texttt{EnqOp} descriptors.
%]]
}
\label{lst:complicated}
\end{figureAsListingWide}

\paragraph{Reading an element.} 
Once a \texttt{dequeue()} operation reads the first element (line~\ref{line:dequeue:readElem}), it cannot simply read cell \texttt{a[d \% C]}, as there can be a successful \texttt{EnqOp} in \texttt{ops} array which covers the cell but has not written the element to the cell yet. 
Thus, we use a special \texttt{readElem} function (lines~\ref{line:readElem:start}--\ref{line:readElem:end}), which goes through array \texttt{ops} looking for a successful descriptor that covers the cell, and returns the corresponding element if one is found. 
If there are no such descriptors, \texttt{readElem} returns the value stored in array \texttt{a}.

Note that we invoke \texttt{readElem} between two reads of counter \texttt{dequeues} that should return identical values for the operation to succeed. 
Thus, we guarantee that the current \texttt{dequeue()} has not been missed the round during the \texttt{readElem} invocation, and, therefore, cannot return a value inserted during one of the next rounds. 
Concurrently, the latest successful \texttt{enqueue(..)} invocation to the corresponding cell might  either be ``stuck'' in the \texttt{ops} array or have  successfully written its value to the cell and completed; \texttt{readElem} finds the correct element in both scenarios.

\paragraph{Adding a new element.}
As discussed above, \texttt{enqueue(..)} creates a new \texttt{EnqOp} descriptor (line~\ref{line:enq_createdesc}), which tries to apply the operation if counter \texttt{enqueues} has not been changed (line~\ref{line:enq_applydesc}). 
The corresponding \texttt{apply} function is described in lines~\ref{line:apply:start}--\ref{line:apply:end}. 
First, it tries to find out if a concurrent operation   covers the corresponding cell in \texttt{a} (line~\ref{line:apply:findOp}). 
If no operation covers the cell, it tries to put \texttt{EnqOp} into \texttt{ops}; this attempt can fail if a concurrent \texttt{enqueue(..)} does the same{---}only one such \texttt{enqueue(..)} should succeed. 
If the current \texttt{EnqOp} is successfully inserted into \texttt{ops} (\texttt{putOp} returns a valid slot number in line~\ref{line:apply:tryPutOp}), the cell is covered and the current thread is eligible to update it, thus, the operation is completed (line~\ref{line:apply:completeOp}). 
Otherwise, if another operation covers the cell, we check if the descriptor belongs to the previous round (line~\ref{line:apply:checkOutdated:start}) and try to replace it with the current one (lines~\ref{line:apply:op_success_true}--~\ref{line:apply:op_success_false}). 
However, since the current thread can be suspended, the found descriptor might belong to the current or a subsequent round, or the replacement in line~\ref{line:apply:findOp} fails because a concurrent \texttt{enqueue(..)} succeeds before us; in this case, the \texttt{enqueue(..)} attempt fails.

An intuitive state diagram for \texttt{ops} slots is presented in Figure~\ref{fig:op_slot}.  
Starting from the initial empty state ($\perp$), an operation that reads the counter value \texttt{e$_i$} and intends to cover cell \texttt{a[e$_i$ \% C]} occupies the slot (the state changes to the ``yellow'' one). 
After that, it (or another helping thread) checks whether there exists an operation descriptor in other slots that already covers the same cell, and fails if so (moving to the ``red'' state), or successfully applies the operation and covers the slot (moving to the ``green'' state). 
At last, the operation writes the value to the cell and frees the slot (moving to the initial state). 
If a next-round \texttt{enqueue(..)} arrives while the cell is still covered, it replaces the old \texttt{EnqOp} descriptor with a new one (moving to the next ``green'' state). 
In this case, the operation is completed by the thread that covers the cell.
Otherwise, we just put $\perp$ instead of the descriptor in \texttt{ops}.
%\pk{Very hard to digest. Confused: does not a slot in ops pass by $\bot$ when the operation completes (line~\ref{line:completeOp:end})?}

%\begin{wrapfigure}{r}{0.7\textwidth}
%    \centering
%    \vspace{-1em}
%    \includegraphics[width=0.7\textwidth]{images/op_slot.pdf}
%    \caption{Life-cycle of the \texttt{ops} array slots, all transitions are performed by successful \texttt{CAS}-s. The states marked with ``OP'' store a descriptor reference, while the color specifies the status of this descriptor: ``yellow'' means ``not defined yet'', ``red'' {---} ``failed'', and ``green'' -- ``success''.}
%    \label{fig:op_slot}
%    \vspace{-1em}
%\end{wrapfigure}

\begin{figure*}
\includegraphics[width=\textwidth]{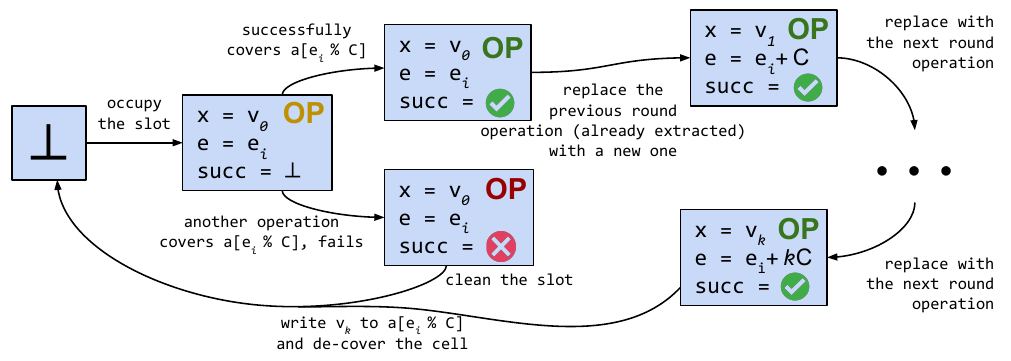}
\caption{Life-cycle of the \texttt{ops} array slots, all transitions are performed by successful \texttt{CAS}-s. The states marked with ``OP'' store a descriptor reference, while the color specifies the status of this descriptor: ``yellow'' means ``not defined yet'', ``red'' {---} ``failed'', and ``green'' -- ``success''.}
\label{fig:op_slot}
\end{figure*}

\paragraph{Atomic \texttt{EnqOp} put.}
Finally, we describe how to atomically change the descriptor status depending on whether the corresponding cell is covered. 
In our algorithm, we perform all such puts sequentially using the \texttt{activeOp} variable that defines the currently active \texttt{EnqOp} (line~\ref{line:activeop}).
To ensure lock-freedom, other threads can help modifying the status of the operation, and only the fastest among them succeeds. 
Hence,  \texttt{putOp} function (lines~\ref{line:putop:start}--\ref{line:putop:end}) circularly goes through the \texttt{ops} slots, trying to find an empty one and to insert the current operation descriptor into it (line~\ref{line:putop:put}). 
At this point, the status of this \texttt{EnqOp} is not defined.
Then the operation should be placed to \texttt{activeOp} via \texttt{startPutOp}, which tries to atomically change the field from $\bot$ to the descriptor (line~\ref{line:startPutOp:set}) and helps other operations if needed (lines~\ref{line:startPutOp:helpStart}--\ref{line:startPutOp:helpEnd}).
Afterwards, the status of the operation is examined in \texttt{tryPut} method in \texttt{EnqOp} class {---} it checks that the cell is not covered (lines~\ref{line:tryput:checkCovered:start}--\ref{line:tryput:checkCovered:start}) and the \texttt{enqueues} counter has not been changed (lines~\ref{line:tryput:checkEnqChanged:start}--\ref{line:tryput:checkEnqChanged:end}). 
At the end, \texttt{activeOp} is cleaned up (line~\ref{line:putop:clean_activeop}) so other descriptors can be processed.
If \texttt{tryPut} succeeds, \texttt{putOp} returns the corresponding slot number; otherwise, it cleans the slot (line~\ref{line:putop:clean}) and returns \texttt{-1}.

\subsection{Progress Guarantee} \label{sec:algo_lf}

We prove lock-freedom of our bounded queue algorithm.
%which matches the lower bound. 
We consider \texttt{dequeue()} and \texttt{enqueue(..)} operations separately.

% We devote all this section to prove the lock-freedom of our algorithm.

% \begin{theorem}
% The algorithm presented at Listing~\ref{lst:complicated} is lock-free.
% \end{theorem}

\paragraph{The \texttt{dequeue()} operation.}
Following the code, the only place where the operation can restart is an unsuccessful \texttt{dequeues} counter increment (line~\ref{line:dequeue:inc}). However, each increment failure indicates that a concurrent \texttt{dequeue()} has successfully performed the same increment and completed. Lock-freedom for dequeues then trivially follows.

\paragraph{The \texttt{enqueue(..)} operation.} 
It is easy to notice that a successful \texttt{enqueues} counter increment indicates the system's progress. This way, if we will show that the following assumption is incorrect: the counter is never changing and all the in-progress enqueues are stuck in a live-lock, we automatically show the \texttt{enqueue(..)} lock-freedom.

Let's denote the current \texttt{enqueues} value, that cannot be changed, as $e_\textit{stuck}$. Our main idea is proving that there should appear a \emph{successful} descriptor with field \texttt{e} equals $e_\textit{stuck}$. In this case, following the code, if the enqueues are stuck in a live-lock, there should exist an infinite number of restarts in \texttt{enqueue(..)} operations. By that these \texttt{enqueue(..)} operations should find the successful descriptor in \texttt{ops} array, if there exists one, and, thus, the counter \texttt{enqueues} should be incremented at some point. Thus, the assumption is correct only if a descriptor with \texttt{successful} field set to \texttt{true} cannot be put into \texttt{ops}. 

Consider the function \texttt{apply} which is called by \texttt{enqueue(..)} operation (line~\ref{line:enq_applydesc}). There are two possible scenarios: if \texttt{findOp} in line~\ref{line:apply:findOp} finds a descriptor or not. Consider the first case. If it finds the descriptor $op$ with $op.e \geq e_\textit{stuck}$ then either we found the required descriptor or \texttt{enqueues} was incremented. Thus the found descriptor is from the previous round, the operation tries to replace it with a new one, already in the successful state (line~\ref{line:apply:replace}). The corresponding \texttt{CAS} failure indicate that either the previous-round operation is completed and the slot becomes $\bot$, or a concurrent \texttt{enqueue(..)} successfully replaced the descriptor with a new one with \texttt{op.e = $e_\textit{stuck}$}. Since the second case breaks our assumption, we consider that the previous-round operation put $\bot$ into the slot, which leads to the second scenario, when we do not find the descriptor since the cell is not covered.

In the second case, \texttt{findOp} in \texttt{apply} function (line~\ref{line:apply:findOp}) did not find a descriptor and tries to put its own. Let's assume now that each thread successfully finds and occupies an empty slot in a bounded number of steps in \texttt{putOp} function (lines~\ref{line:putop:occupy:start}--\ref{line:putop:occupy:end}). Thus, some descriptor with \texttt{op.e = $e_\textit{stuck}$} should be successfully set into \texttt{activeOp} (see function \texttt{startPutOp}), and the following \texttt{tryPut} invocation sets the status to the successful one since there is no other successful descriptor that covers the cell according to the main assumption. This means, that the only way not to break this assumption is that  is to never occupy a slot in \texttt{putOp}. However, since the \texttt{ops} size equals to the number of processes, the only way not to cover a slot during an array traversal, is that there exist an infinite amount of \texttt{EnqOp} descriptors which are put into \texttt{ops}. Obviously, \texttt{e} field on all of them is equal to $e_\textit{stuck}$ and at least one of them passes to \texttt{activeOp} and become successful by using the same argument as earlier.

As a result, we show that there should occur a successful descriptor with \texttt{op.e = $e_\textit{stuck}$} in \texttt{ops} array in any possible scenario, which breaks the assumption and, thus, provides the lock-freedom guarantee for \texttt{enqueue(..)}.

\subsection{Correctness} \label{sec:proof}
% We devote all this section to prove the linearizability of our algorithm.
% \begin{theorem}
% The algorithm presented at Listing~\ref{lst:complicated} is linearizable.
% \end{theorem}
This subsection is devoted for the linearizability proof of our algorithm.
We split this subsection into two parts: the intuition of the proof and the full version.

\subsubsection{Intuition} 
For \texttt{dequeue()}, lock-freedom guarantee is immediate,  as the only case when the operation fails and has to retry is when the \texttt{dequeues} counter was concurrently incremented, which indicates that a concurrent \texttt{dequeue()} has made progress.

As for \texttt{enqueue(..)}, all the \texttt{CAS} failures except for when the \texttt{enqueues} counter increments (line~\ref{line:enqueue:inc} in \texttt{enqueue(..)} and~\ref{line:enqueue:inc:1} in \texttt{completeOp}) or the \texttt{ops} slot occupation (line~\ref{line:putop:put}) fail intuitively indicate the system's progress. 
%
%[[PK trivial consequence
%This way, it is totally fine when threads retry their operations (or parts of them) in case of these failures.
%]]
In general, an \texttt{enqueue(..)} fails only due to helping, which does not cause retries. 
The only non-trivial situation is when a thread is stuck while trying to occupy a slot in \texttt{putOp}. Since the \texttt{ops} size equals the number of processes, it is guaranteed that when a thread intents to put an operation descriptor into \texttt{ops}, there is at least one free slot. 
Thus, if no slot can be taken during a traversal of \texttt{ops}, \texttt{enqueues} by other processes successfully take slots, and the system as a whole still makes progress. 

The linearization points of the operations can be assigned as follows. 
A successful \texttt{dequeue()} operation linearizes at successful \texttt{CAS} in line~\ref{line:dequeue:inc}.
A failed \texttt{dequeue()}  operation linearizes in line~\ref{line:deq:read:e}. 
For a successful \texttt{enqueue(..)} operation, we consider the descriptor \texttt{op} (created by that operation) that appears in \texttt{ops} array and has its \texttt{successful} field set. 
The linearization point of the operation is in \texttt{CAS} that changes \texttt{enqueues} counter from \texttt{op.e} to \texttt{op.e + 1}.
A failed \texttt{enqueue(..)} operation linearizes in line~\ref{line:enq:read:d}. 
One can easily check that queue operations ordered according to their linearization points constitute a correct sequential history. 

\subsubsection{Full proof}
At first, we suppose that all operations during execution are successful so the checks for emptiness and fullness never satisfy. This way, we are provided with an arbitrary finite history $H$, and we need to construct a linearization $S$ of $H$ by assigning a linearization point to each completed operation $\delta$ in $H$.
We prove several lemmas before providing these linearization points.

% The next lemma, we will be using in our proof without mentioning.
\begin{lemma}
At any time during the whole execution, it is guaranteed that $\texttt{dequeues} \leq  \texttt{enqueues} \leq \texttt{dequeues} + C$.
\end{lemma}
\begin{proof}
We prove this by contradiction. Let \texttt{dequeues} exceed \texttt{enqueues}. Since the counters are always increased by one, we choose the first moment when \texttt{dequeues} becomes equal \texttt{enqueues + 1} which, obviously, happens at the corresponding \texttt{CAS} invocation $\delta_1$ in \texttt{dequeue()} (line~\ref{line:dequeue:inc}). Consider the situation right before this \texttt{CAS} $\delta_1$: \texttt{enqueues} and \texttt{dequeues} were equal. Consider the previous successful \texttt{CAS} $\delta_0$ that incremented \texttt{enqueues} or \texttt{dequeues} counter. Obviously, it was a \texttt{CAS} on \texttt{dequeues} since otherwise we chose not the first moment in the beginning of the proof. Thus, between $\delta_0$ and $\delta_1$ \texttt{enqueues} and \texttt{dequeues} did not change, and the emptiness check at line~\ref{line:hard:deq_emptiness} should succeed during the \texttt{dequeue()} invocation that has performed $\delta_1$, so the \texttt{CAS} $\delta_1$ cannot be executed.

The situation when \texttt{enqueues} exceed \texttt{dequeues + C} can be shown to be unreachable in a similar manner.
\end{proof}

Further, we say that a descriptor \texttt{EnqOp} is \emph{successful} if at some moment its \texttt{successful} field was set to \texttt{true} and it was in \texttt{ops} array.

\begin{lemma}
\label{lem:unique:per:e}
During the execution, for each value $pos \in [0, \texttt{enqueues}]$ there existed exactly one successful \texttt{EnqOp} object with the field \texttt{e} set to the specified $pos$.
\end{lemma}
\begin{proof}
At first, we show that there exists at least one \texttt{EnqOp} descriptor.
The \texttt{enqueues} counter can never increment if there does not exist the corresponding \texttt{EnqOp} object: \texttt{CAS} at line~\ref{line:enqueue:inc} is successful only if we put new object or the object exists, the same goes for line~\ref{line:enqueue:inc:1}.

This object is unique, since we add objects \texttt{EnqOp} into array \texttt{ops} one by one using \texttt{startPutOp} function. In other words, the checks in Lines~\ref{line:tryput:checkCovered:start}-\ref{line:tryput:checkEnqChanged:end} are performed in the serialized manner. Thus, the object cannot become successful if it finds another object with the same value in field \texttt{i} in line~\ref{line:tryput:checkCovered:start}.
\end{proof}

We denote the last \texttt{EnqOp} created by an \texttt{enqueue(..)} operation in line~\ref{line:enqop:create} as the \emph{last descriptor} of the operation.

\begin{lemma}
\label{lem:unique:per:op}
The last descriptor of \texttt{enq($x$)} either becomes successful or the operation never finishes. Also, exactly one \texttt{EnqOp} descriptor created by the operation becomes successful.
\end{lemma}
\begin{proof}
We note that if the descriptor does not become successful then the operation restarts at line~\ref{line:enqueue:successful:check}.
Note that the flag \texttt{successful} cannot be \texttt{true} at line~\ref{line:enqueue:successful:check} if the object was not added to \texttt{ops} array.
Thus, it should become successful.
\end{proof}

Since by Lemma~\ref{lem:unique:per:e} for each value $pos$ there exists exactly one \texttt{EnqOp} descriptor with the specified $pos$ and by Lemma~\ref{lem:unique:per:op} \texttt{enqueue(..)} operation creates only one successful description, we can make a bijection between \texttt{enqueue(..)} operations and the positions $pos$: that are chosen as the value of the field \texttt{e} field in the last descriptor of the operation. (Note, that there is an obvious bijection between operations and descriptors.) Thus, we can say ``the position of the enqueue'' and ``the enqueue of the position''. Also, we say that \texttt{EnqOp} descriptor $op$ covers a cell $id$ if \texttt{op.i = id}.

The linearization points are defined straightforwardly. For \texttt{deq()} the linearization point is at the successful \texttt{CAS} performed in line~\ref{line:dequeue:inc}. The linearization point for $\pi=\texttt{enq($x$)}$ is the successful \texttt{CAS} performed on \texttt{enqueues} counter in line~\ref{line:enqueue:inc:1} or line~\ref{line:enqueue:inc} from $E$ to $E + 1$ where $E$ is the position of $\pi$.

Let $S_{\sigma}$ be the linearization of the prefix $\sigma$ of the execution $H$. Also, we map the state of our concurrent queue after the prefix $\sigma$ of $H$ to the state of the sequential queue $Q_{\sigma}$ provided by the algorithm in Figure~\ref{fig:sequential}. This map is defined as follows: the values of the counters \texttt{enqueues} and \texttt{dequeues} are taken as they are, and for each $pos \in [\texttt{dequeues}, \texttt{enqueues})$ the value in the position \texttt{$pos$ \% C} contains $op.x$ of successful \texttt{EnqOp} $op$ with $op.e = pos$ from \texttt{ops} array or, if such $op$ does not exist, simply \texttt{a[$pos$ \% C]}.

% The following lemma is enough to prove the correctness of our linearization.
\begin{lemma}
\label{lem:specification}
The sequential history $S_H$ complies with the \texttt{queue} specification.
\end{lemma}
\begin{proof}
We use induction on $\sigma$ to show that for every prefix $\sigma$ of $H$, $S_{\sigma}$ is a queue history and the state of the sequential queue $L_{\sigma}$ from Figure~\ref{fig:sequential} after executing $S_{\sigma}$ is equal to $Q_{\sigma}$. The claim is clearly true for $\sigma = \epsilon$. For the inductive step we have $\sigma=\sigma_1 \circ \delta$ and $L_{\sigma_1}$ after $S_{\sigma_1}$ coincides with $Q_{\sigma_1}$.

Suppose, $\delta$ is the successful \texttt{CAS} in line~\ref{line:dequeue:inc} of operation $\pi = \texttt{deq()}$. Consider the last two reads \texttt{d := dequeues}, $\delta_1$, and \texttt{e := enqueues}, $\delta_2$, in Line~\ref{line:hard:deq_dc_start} in $\pi$. Suppose that $\delta_1$ reads $D$ as \texttt{d} and $\delta_2$ reads $E$ as \texttt{e}. Let $\sigma_2$ be the prefix of $H$ that ends on $\delta_2$. Since we forbid $\pi$ to satisfy the emptyness property we could totally say that the linearization point of the \texttt{enqueue(..)} of position $D$ already passed since $E > D$, otherwise, $\pi$ should have been restarted in line~\ref{line:hard:deq_dc_end}. Thus, by induction $Q_{\sigma_2}$ and $L_{\sigma_2}$ both contains same $v$ at position \texttt{d \% C}. The same holds for $Q_{\sigma_1}$ and $L_{\sigma_1}$. Thus, between $\delta_2$ and $\delta$ queue $Q$ always contains $v$ in \texttt{d \% C} such as queue $L$. By the definition of $Q$ this means during that interval there is always exist either successful \texttt{EnqOp} $op$ with $op.e = D$ and $op.x = v$ or $v$ is stored in \texttt{a[$D$ \% C]}, thus \texttt{readElem(d \% C)} in line~\ref{line:dequeue:readElem} would correctly read $v$. This means, that the result of \texttt{deq()} in both, the concurrent and the sequential, queues coincides.

Suppose, $\delta$ is the successful \texttt{CAS} in line~\ref{line:enqueue:inc} of operation $\pi = \texttt{enq($v$)}$.
Suppose that during \texttt{e := enqueues} in line~\ref{line:hard:enq_dc_start}, $\delta_1$, we read $E$.
We want to prove that three things: 1)~there already existed successful \texttt{EnqOp} $op$ with $op.e = E$; 2)~there is no \texttt{EnqOp} $op$ with \texttt{op.e < $E$} and \texttt{op.i = $E$ \% C} in \texttt{ops} array; and 3)~none of \texttt{EnqOp} descriptor $op_2$ with \texttt{$op_2$.e > $E$} and \texttt{$op_2$.i = $E$ \% C} exist before $\delta$.
It will be enough since by our mapping the cell \texttt{$E$ \% C} will contain the same value in both queues $Q_{\sigma}$ and $L_{\sigma}$, because if there is no successful descriptor $op$ in \texttt{ops} with \texttt{op.e = $E$} then it already dumps its value to array \texttt{a} in line~\ref{line:completeOp:dump} and nobody could have overwritten in.

Note, that the second and the third part can be proved very simply. As for the second part, if we find an outdated descriptor with the same \texttt{i} we replace it to the new one in line~\ref{line:apply:replace}. As for the third part, since \texttt{enqueues} does not exceed $E$ before $\delta$, then no process can even create \texttt{EnqOp} descriptor $op$ with \texttt{op.e > $E$}.

This means, that we are left to prove the first part. Since the \texttt{CAS} in line~\ref{line:enqueue:inc}, $\delta$, is successful then \texttt{enqueues} counter did not change in-between $\delta_1$ and $\delta$. Thus, during \texttt{apply} operation in line~\ref{line:enq_applydesc} \texttt{enqueues} was constant. Let us look on what this operation is doing. At first, it tries to find a successful descriptor \texttt{EnqOp} $op$ for which \texttt{$op = E$ \% C} in line~\ref{line:apply:findOp}. If it does not exist (the check in line~\ref{line:apply:checkCovered} succeeds) it tries to put its own descriptor by using the function~\texttt{putOp} in line~\ref{line:apply:tryPutOp}. As the first step, it tries to put not yet successful \texttt{EnqOp} $op$ into array \texttt{ops} (line~\ref{line:putop:put}). The operation did not become successful in \texttt{tryPut()} function only if either there exists \texttt{EnqOp} $op$ with \texttt{$op = E$ \% C} or \texttt{enqueues} changes. Since, \texttt{enqueues} cannot change as we discussed prior, thus in line~\ref{line:tryput:checkCovered:start} finds \texttt{EnqOp} $op$ with \texttt{$op.i = E$ \% C}. Note that since \texttt{enqueues} did not change and \texttt{findOp} in line~\ref{line:apply:findOp} did not find an operation, then $op.e$ has to be equal to $E$, thus, providing with what we desired.
In the other case, \texttt{apply} finds \texttt{EnqOp} $op$ with \texttt{$op.i = E$ \% C}. Note that since \texttt{enqueues} does not change $op.i$ cannot exceed $e$. Thus, if the check in line~\ref{line:apply:checkOutdated:start} succeeds, $op.e$ should be equal to $E$ and we are done. Otherwise, this means that we found some old $op$ and we try to replace it by \texttt{CAS} in line~\ref{line:apply:replace}. Due to the fact that \texttt{enqueues} does not change, \texttt{CAS} can fail only if another \texttt{EnqOp} descriptor $op$ with $op.e = E$ is successfully placed. By that, there existed \texttt{EnqOp} object that we desired. Thus, after we add $\delta$ to $\sigma_1$, $Q_{\sigma}$ and $L_{\sigma}$ coincides.

Suppose, $\delta$ is the successful \texttt{CAS} in line~\ref{line:enqueue:inc:1} of operation $\pi = \texttt{enq($v$)}$. As in the previous case, it is enough to prove two things: 1) there already exists corresponding \texttt{EnqOp}; 2) there are no other \texttt{EnqOp} that cover the same cell in \texttt{ops} array. The second statement is easy to prove as before. For the first one, we note that we apply $\delta$ only if we find a successful descriptor by \texttt{readOp(opSlot)} in line~\ref{line:completeop:readop}.

Finally, suppose that $\delta$ is none of the successful \texttt{CAS} in Lines~\ref{line:dequeue:inc}, \ref{line:enqueue:inc} and \ref{line:enqueue:inc:1}. Thus, $L_{\sigma}$ does not change between $\sigma_1$ and $\sigma$. At the same time, $Q_{\sigma}$ can change only due to the removal of \texttt{EnqOp} descriptor from array \texttt{ops}. This can happen in two places: successful \texttt{CAS} operations in line~\ref{line:apply:replace} and in line~\ref{line:com}. In the first case, nothing happens since we simply replace old \texttt{EnqOp} descriptor with the new one, thus $Q_\sigma$ remains the same. In the second case, the successful \texttt{CAS} means that no new descriptor is stored in array \texttt{ops} with the same value in field~\texttt{i} and \texttt{enqueues} does not yet come to the next round, otherwise, $op$ read in line~\ref{line:completeop:readop} should have been replaced by \ref{line:apply:replace}. Thus, the value stored in array \texttt{a} is exactly what was stored by $op$ until \texttt{dequeues} passes through that position.
\end{proof}

Note that the proof of the previous lemma is enough to show that the implementation is linearizable if we throw away operations for which check for emptiness or fullness is satisfied. However, such operations are very easy to linearize. Unsuccessful \texttt{deq()} is linearized at line~\ref{line:deq:read:e} while unsuccessful \texttt{enq($v$)} is linearized at line~\ref{line:enq:read:d}. Note that the correctness of the choice of the linearization points can be proved the same way as above, since \texttt{enqueues} and \texttt{dequeues} counters of $Q_{\sigma}$ match\texttt{enqueues} and \texttt{dequeues} counters of $L_{\sigma}$.

\end{document}